\documentclass[11pt,a4paper,final]{amsart}

\newtheorem{theorem}{Theorem}

\newtheorem{example}{Example}

\newtheorem{remark}{Remark}

\usepackage{amsmath}
\usepackage{amsfonts}
\usepackage{amssymb}
\usepackage{graphicx}
\usepackage{comment}
\usepackage{bm}
\usepackage{color} 

\def\bmtheta{{\bm{\theta}}}
\def\bmSigma{{\bm{\Sigma}}}
\def\bmmu{{\bm{\mu}}}

\newcommand{\R}{\mathbb{R}}

\newcommand{\transp}{{T}}
\renewcommand{\d}{\hbox{d}}

\begin{document}

\title[]{Dynamic intertemporal utility optimization by means of Riccati transformation of
Hamilton-Jacobi Bellman equation}

\author{So\v{n}a Kilianov\'a and Daniel \v{S}ev\v{c}ovi\v{c}
}

\address{Department of Appl. Mathematics and Statistics, FMFI, Comenius University in Bratislava, 842 48 Bratislava, Slovakia, {\tt kilianova@fmph.uniba.sk} (S.~Kilianov\'a), {\tt sevcovic@fmph.uniba.sk} (D.~\v{S}ev\v{c}ovi\v{c})
}

\begin{abstract}
In this paper we investigate a dynamic stochastic portfolio optimization problem involving both the expected terminal utility and intertemporal utility maximization. We solve the problem by means of a solution to a fully nonlinear evolutionary Hamilton-Jacobi-Bellman (HJB) equation. We propose the so-called Riccati method for transformation of the fully nonlinear HJB equation into a quasi-linear parabolic equation with non-local terms involving the  intertemporal utility function. As a numerical method we propose a semi-implicit scheme in time based on a finite volume approximation in the spatial variable. By analyzing an explicit traveling wave solution we show that the numerical method is of the second experimental order of convergence. As a practical application we compute optimal strategies for a portfolio investment problem motivated by market financial data of German DAX 30 Index and show the effect of considering intertemporal utility on optimal portfolio selection.

\medskip
\noindent AMS-MOS Classification: {35K55, 34E05, 70H20, 91B70, 90C15, 91B16}


\end{abstract}

\maketitle

\section{Introduction}

In this paper, we investigate the impact of presence of a nontrivial intertemporal utility function on stochastic optimal portfolio selection problem. The problem can be formulated in terms of the expected terminal and intertemporal utility maximization problem, in which the underlying stochastic process is controlled by a vector of time-dependent weights of assets entering a financial portfolio. 

To solve the expected terminal and intertemporal utility maximization problem, we follow a methodology based on solving a fully nonlinear parabolic Hamilton-Jacobi-Bellman equation for the intermediate value function of the corresponding optimal control problem. A similar problem was investigated by Federico, Gassiat and Gozzi in \cite{Federico}, where they studied a problem of terminal and intertemporal utility maximization in an investment-consumption portfolio setting and the current utility being dependent also on the wealth process. They studied properties of solutions to a dual control problem.

The novelty of our paper is generalization of the transformation method proposed and analyzed by  Abe and Ishimura \cite{AI}, Ishimura and \v{S}ev\v{c}ovi\v{c} \cite{IshSev} and Kilianov\'a and \v{S}ev\v{c}ovi\v{c} \cite{KilianovaSevcovicANZIAM, KilianovaSevcovicKybernetika} for the case of a nontrivial intertemporal utility function. The transformation is also referred to as the Riccati transformation as it involves the ratio between the second and the first derivatives of the value function. The transformed function can be viewed as the absolute risk aversion coefficient of an investor. Secondly, we generalize the underlying stochastic process to more general processes with arbitrary drift and volatility functions. Such a general setting can include, in particular, processes arising in the so-called  worst-case portfolio optimization studied recently  by Ki\-lia\-no\-v\'a and Trnovsk\'a in \cite{KilianovaTrnovska}. In contrast to the problem involving the terminal utility maximization only (cf. \cite{AI},\cite{IshSev}, \cite{KilianovaSevcovicANZIAM}, \cite{KilianovaSevcovicKybernetika}), the resulting transformed equation is a non-local quasi-linear parabolic equation containing non-local terms involving the  intertemporal utility function. The non-local parabolic equation can be further  transformed into a coupled system of two quasi-linear local parabolic equations. We analyze these governing equations and show how their solutions are related to solving the original HJB equation. As a tool for solving the associated terminal and intertemporal utility maximization problem, we generalize the numerical method proposed by  Kilianov\'a and \v{S}ev\v{c}ovi\v{c} \cite{KilianovaSevcovicANZIAM, KilianovaSevcovicKybernetika} for the case when a non-local term appears in the quasi-linear parabolic equation. We furthermore derive a-priori lower and upper bounds of a solution which are given in terms of the risk aversion coefficients of the terminal and intertemporal utility functions. The main advantage of the Riccati transformation method is twofold. First, the transformed function has a practical representation and meaning as an intertemporal risk aversion of an investor and it is a  globally bounded function even in the case when the utility function is unbounded. Moreover, there are natural boundary conditions for a solution defined on  a truncated numerical domain. Secondly, the nonlinearity appearing as a diffusion function in the transformed equation can be computed in a fast and efficient way using modern tools of conic convex programming. 

As a practical application we compute optimal strategies for a portfolio investment problem motivated by the market financial data from German DAX 30 Index. We compare the optimal portfolio selection strategies for the case of absence of intertemporal utility maximization to the case when a non-trivial intertemporal utility is considered. We illustrate the effect of an intertemporal utility function  on the optimal portfolio selection. 

The paper is organized as follows. In the next section we introduce and discuss basic model assumptions made on the underlying stochastic process. The process of logarithmized portfolio wealth $x_t$ at time $t$ is controlled by a vector of weights $\bmtheta_t$ belonging to a compact convex subset of $\R^n$. A dynamic stochastic optimization problem with intertemporal utility is formulated in Section 3. Following the Bellman optimality principle, we present a fully nonlinear backward parabolic equation for the intermediate value function satisfying a given terminal condition. In Section 4 we present the so-called Riccati transformation of the value function, leading us from the fully nonlinear HJB equation to a single  quasi-linear parabolic equation in the divergent form containing a non-local term. We furthermore analyze qualitative properties of an auxiliary value function arising from a parametric convex programming problem. Existence of a classical H\"older smooth solution and its a-priori bounds are also derived in this section.
Section 5 is devoted to a numerical approximation scheme for solving the transformed quasi-linear parabolic function. The scheme is based on the finite volume approximation method involving dual finite volumes. We compare the numerical scheme to the fixed policy iteration method for solving HJB equations as investigated by e.g. by Huang et al. \cite{Huang2010},  Reisinger and Witte \cite{Reisinger}. Finally, in Section 6 we test the accuracy of the proposed numerical method on an explicit traveling wave example and compute the experimental order of convergence. We show that the experimental order of convergence is approximately two which indicates the second order of convergence of the numerical method. Subsequently, we apply the proposed method to optimal portfolio selection problem and present corresponding results with and without intertemporal utility maximization.  

\section{Underlying stochastic process with control} 

Throughout the paper we will assume  that the underlying stochastic process $\{x_t\}$ satisfies the stochastic differential equation (SDE)
\begin{equation}
\label{process_x}
d x_t = \mu(x_t, t, \bmtheta_t) dt + \sigma(x_t, t, \bmtheta_t) dW_t\,,
\end{equation}
where the control process $\{\bmtheta_t\}$ is adapted to the process $\{x_t\}$, $\{W_t\}$ is the standard one-dimensional Wiener process and functions $(x,t,\bmtheta)\mapsto\mu(x,t,\bmtheta)$ and $(x,t,\bmtheta )\mapsto\sigma(x,t,\bmtheta)^2$ are $C^{1,1}$ smooth in $x, t$ and $\bmtheta$ variable, i.e. their first derivatives are Lipschitz continuous functions.

\begin{remark}
The motivation for studying the SDE (\ref{process_x}) as an underlying stochastic process controlled by $\bmtheta$ arises from the stochastic dynamic optimal portfolio management. Let $x_t^i = \ln y_t^i$ denote the logarithm
of the asset value $y_t^i$ entering a portfolio consisting of $n$ assets with vector of weights $\bmtheta$. Then $d x_t^i = d y_t^i / y_t^i$ is the return on the asset $i$. Suppose that the process of each such a return is driven by
\[
d x_t^i = \mu^i dt + \sum_{k=1}^n \sigma^{ki} d W_t^k
\]
where $W_t^j$ is a one-dimensional Wiener process such that the increments $dW_t^j$ and $dW_t^i$ are independent for $j\not=i$. The mean return on the increment of the portfolio $x^\bmtheta=\sum_{i=1}^n \theta^i x^i$ with the vector of weights $\bmtheta$ is $\bmmu^T \bmtheta \, dt$ and its variance is equal to $\sum_{i,j,k=1}^n \theta^i \sigma^{ki}\sigma^{kj} \theta^j \, dt$.

Following Merton \cite{Merton1, Merton2} we can describe the stochastic process $x_t^\bmtheta$ by the following one-dimensional SDE of the form (\ref{process_x}):
\[
d x_t^\bmtheta = \bmmu^T \bmtheta \, dt + \sigma(\bmtheta) dW_t
\] 
where $W_t$ is the one-dimensional Wienner process, $\sigma(\bmtheta)^2 = \bmtheta^T \bmSigma  \bmtheta$ and $\bmSigma$ is the covariance matrix, $\bmSigma_{ij}= \sum_{k=1}^n \sigma^{ki}\sigma^{kj}$.
\end{remark}

\begin{example}
As an example of the stochastic process (\ref{process_x}), one can consider a portfolio optimization problem with regular cash inflow (e.g. pension planning). In this example, the volatility function is given by
\begin{equation}
\sigma(x,t,\bmtheta)^2 =\bmtheta^T \bmSigma  \bmtheta,
\label{volatility}
\end{equation}
where $\bmSigma$ is a positive definite covariance matrix of asset returns. The drift function is given by
\begin{equation}
\mu(x,t,\bmtheta) = \bmmu^T\bmtheta - \frac12 \sigma(x,t,\bmtheta)^2  +\varepsilon e^{-x} + r,
\label{drift}
\end{equation}
where $\bmmu$ is the vector of mean returns of assets, $\varepsilon$ is an inflow ($\varepsilon>0$)/outflow ($\varepsilon<0$) to/from the portfolio, $r\ge0$ is an interest rate of a risk-free bond. The stochastic process $\{x^\bmtheta_t\}$ controlled by $\{\bmtheta_t\}$ is a logarithmic transformation of the stochastic process $\{y_t^{\tilde\bmtheta} \}_{t\ge0}$ driven by the stochastic differential equation
\begin{equation}
d y_t^{\tilde\bmtheta} = (\varepsilon + (r +\mu({\tilde\bmtheta})) y_t^{\tilde\bmtheta})
d t + \sigma(\tilde\bmtheta) y_t^{\tilde\bmtheta} d W_t, \label{processYeps}
\end{equation}
where $\tilde\bmtheta(y,t) = \bmtheta(x,t)$ with $x=\ln y$  (cf. Kilianov\'a and {\v S}ev{\v c}ovi{\v c} \cite{KilianovaSevcovicANZIAM}).

Another example stems from the so-called worst-case portfolio optimization problem investigated by Kilianov\'a and Trnovsk\'a \cite{KilianovaTrnovska}. The volatility function is given by
\[
\sigma(x,t,\bmtheta)^2 = \max_{\bmSigma\in{\mathcal K}}\bmtheta^T \bmSigma  \bmtheta, 
\]
where $\mathcal K$ is an uncertainty convex set of positive definite covariance matrices. Typically, only a part of the covariance matrix is exactly prescribed while other entries are not precisely determined. For instance, if only the diagonal $d$ is known, we have ${\mathcal K}=\{\bmSigma\succ0, \ diag(\bmSigma)=d\}$. The drift function is given by
\[
\mu(x,t,\bmtheta) = \min_{\bmmu\in{\mathcal E}}\bmmu^T\bmtheta - \frac12 \sigma(x,t,\bmtheta)^2  +\varepsilon e^{-x} +r,
\]
where $\mathcal E$ is a given uncertainty convex set of mean returns. 
    
\end{example}

\section{Dynamic stochastic optimization problem with intertemporal utility function}

Our goal is to extend the model of terminal utility maximization studied previously in Kilianov\'a and {\v S}ev{\v c}ovi{\v c} \cite{KilianovaSevcovicANZIAM} by including an intertemporal utility function. Maximization of dynamic utility has been investigated in vast literature in the past by a number of methods. In this paper, we assume that the investor has a certain utility $c$ from intertemporal wealth but a different utility $u$ from terminal wealth. We assume the overall utility to be time-additive. Then we can formulate the problem of dynamic utility maximization as follows:
\begin{equation}
\max_{\bmtheta|_{[0,T)}} \mathbb{E}
\left[u(x_T^\bmtheta)\, + \int_0^T  c(x_s^\bmtheta, s) ds \big| \, x_0^\bmtheta=x_0 \right],
\label{maxproblem}
\end{equation}
(c.f. \cite{Federico} where they included consumption as well). Here $\{x_t^{\bmtheta}\}$ is It\=o's stochastic process of the form (\ref{process_x}) on a finite time horizon $[0,T]$, $u: \mathbb{R} \to \mathbb{R}$ is a given terminal utility function and $x_0$ a given initial state condition of $\{x_t^\bmtheta\}$ at $t=0$. The function $\bmtheta:   \mathbb{R} \times [0,T) \to \R^n$ maps $(x,t) \mapsto \bmtheta(x,t)$ and it represents an unknown control function governing the underlying stochastic process $\{x_t^\bmtheta\}$. The function $c:\mathbb{R} \times [0,T)\to \mathbb{R}$  is the intertemporal utility function. In what follows, we will assume $c$ is a $C^2$ smooth function and it is non-decreasing in the $x$ variable. Clearly, one can add a time discounting factor to both utility functions in (\ref{maxproblem}). We assume that the control parameter $\bmtheta$ belongs to a closed convex subset $\Delta$ of the compact simplex $\mathcal{S}^n = \{\bmtheta \in \mathbb{R}^n\  |\  \bmtheta \ge \mathbf{0}, \mathbf{1}^\transp \bmtheta = 1\} \subset \mathbb{R}^n$, where $\mathbf{1} = (1,\cdots,1)^\transp \in \mathbb{R}^n$.

If we introduce the value function
\begin{equation}
V(x,t):= \sup_{  \bmtheta|_{[t,T)}} 
\mathbb{E}\left[u(x_T^\bmtheta) + \int_t^T c(x_s^\bmtheta, s) ds| x_t^\bmtheta=x \right]
\end{equation}
then $V(x,T):=u(x)$. Following Bertsekas \cite{Bertsekas}, the value function $V=V(x,t)$ satisfies the fully nonlinear Hamilton-Jacobi-Bellman (HJB) parabolic equation: 
\begin{eqnarray}
&& \partial_t V + \max_{ \bmtheta \in \Delta} 
\left(
\mu(x,t,\bmtheta)\, \partial_x V 
+ \frac{1}{2} \sigma(x,t,\bmtheta)^2\, \partial_x^2 V \right)  + c(x,t) = 0\,, 
\nonumber
\\
&& V(x,T)=u(x),  \label{eq_HJB}
\end{eqnarray}
for $ (x,t)\in\R\times [0,T)$; see also \cite{KilianovaSevcovicANZIAM} and Kossaczk\'y, Ehrhardt, G\"unther \cite{KossaczkyEQ, KossaczkyNMTMA}.

As an example of the terminal utility function, one can consider e.g. a constant absolute risk aversion (CARA) function:
\[
u(x) = - e^{-a x},
\]
with constant absolute risk aversion $a\equiv a(x)>0$, where 
\[
a(x) = -\frac{u''(x)}{u'(x)} \quad \hbox{for}\ x\in\R.
\]
We note that a CARA utility function in the variable $x$ corresponds to a CRRA (constant relative risk aversion) utility function $u(y) = -y^{-a}$ in the variable $y=e^x$. In practical applications, $y$ can stand for a portfolio value and $x=\ln y$ its log-transform, for which \eqref{process_x} holds.

Another choice for the utility function $u$ can be a decreasing/increasing absolute risk aversion (DARA/IARA) function with $a(x)$ decreasing/increasing in the $x$ variable. Typically, the intertemporal utility function $c$ is a non-decreasing concave discounted function, i.e. 
\begin{equation}
\label{eq:consumption}
c(x,t) =  -\kappa e^{-d x - \varrho(T-t)},
\end{equation}
where $\kappa, d\ge 0$ and $\varrho$ is a discounting factor. We note that including a discount factor $e^{-rT}$ into the terminal utility function $u$ does not play any role in the solution, as one can transfer this constant into the coefficient $\kappa$ of the intertemporal utility function $c$ simply by multiplying \eqref{maxproblem} by $e^{rT}$.

\section{The Riccati transformation of the HJB equation with  intertemporal utility function}
\label{sec:HJB}

Following the papers by Abe and Ishimura \cite{AI}, Ishimura and \v{S}ev\v{c}ovi\v{c} \cite{IshSev} and Kilianov\'a and \v{S}ev\v{c}ovi\v{c} \cite{KilianovaSevcovicANZIAM}, the Riccati transformation of the value function $V$ can be introduced as follows:

\begin{equation}
\varphi(x,\tau) = - \frac{\partial_x^2 V(x,t)}{\partial_x V(x,t)}, \quad\hbox{where}\ \ \tau=T-t.
\label{eq_varphi}
\end{equation}

Suppose for a moment that the value function $V(x,t)$ is increasing in the $x$-variable. This is a natural assumption in the case when the terminal utility function $u(x)$ is increasing itself. Then the HJB equation (\ref{eq_HJB}) can be rewritten as follows:
\begin{equation}
\partial_t V - \alpha(\cdot,\varphi) \partial_x V + c = 0, \qquad V(\cdot ,T)=u(\cdot),
\label{eq_HJBtransf}
\end{equation}
where $\alpha(x,\tau,\varphi)$ is the value function of the following parametric optimization problem:
\begin{equation}
\alpha(x,\tau,\varphi) = \min_{ \bmtheta \in \Delta} 
\left(
-\mu(x,t,\bmtheta) +  \frac{\varphi}{2}\sigma(x,t,\bmtheta)^2\right), \quad \tau=T-t\,.
\label{eq_alpha_def}
\end{equation}

The proof of the following result is a straightforward generalization of \cite[Theorem 4.1]{KilianovaSevcovicANZIAM} for a more general drift and volatility functions depending on $x$ and $t$ and therefore it is omitted.

\begin{theorem}\label{smootheness}
Assume that the functions $(x,t,\bmtheta)\mapsto\mu(x,t,\bmtheta)$ and $(x,t,\bmtheta )\mapsto\sigma(x,t,\bmtheta)^2$ are $C^{1,1}$ smooth in $x, t$ and $\bmtheta$ variables, and such that the objective function 
$f(x,t,\varphi, \bmtheta):= - \mu(x,t,\bmtheta) + \frac{\varphi}{2} \sigma(x,t,\bmtheta)^2$
is strictly convex in the variable $\bmtheta\in\Delta$ for any  $\varphi\in(\varphi_{min}, \infty)$ where $\Delta\subset\mathbb{R}^n$ is a compact convex set. Then the value function $\alpha$ is $C^{1,1}$ smooth for $x\in\R, \tau\in[0,T], \varphi\in(\varphi_{min},\infty)$. Moreover, $\varphi\mapsto\alpha(\cdot,\varphi)$ is a strictly increasing function. For the derivative of $\alpha$ we have $\alpha^\prime_\varphi(x,\tau,\varphi) = (1/2)\sigma(x,T-\tau,\hat\bmtheta(x,\tau, \varphi))^2$ where $\hat\bmtheta(x,\tau, \varphi)\in\Delta$ is the argument of the minimum of $\alpha(x,\tau, \varphi)$ with respect to $\bmtheta$.
\end{theorem}

\begin{example}\label{examplealpha}
If we consider an example of the decision set $\Delta=\{\bmtheta\in\mathbb{R}^2,\,  \theta_1,\theta_2\ge0, \theta_1+\theta_2=1\}, n=2, \mu=\bmmu^T\bmtheta, \sigma^2=\bmtheta^T\bmSigma\bmtheta$, then the value function $\alpha$ has the form:
\[
\alpha(\varphi)=\left\{ 
\begin{array}{ll}
A \varphi - \frac{B}{\varphi} + C,     & \hbox{if} \ \varphi>\varphi_*, \\
E \varphi + D,     & \hbox{if}\ \varphi\le \varphi_*,
\end{array}
\right.
\]
where constants $A>0, B>0, C, D, E>0, \varphi_*>0$  depend on the mean return vector $\bmmu$ and covariance matrix $\bmSigma$ and are such that $\alpha$ is $C^{1,1}$ continuous function having one point of discontinuity of the second derivative $\alpha''$ at $\varphi_*$. The minimizer $\hat\bmtheta=\hat\bmtheta(\varphi)$ increases the number of positive weights when $\varphi$ passes through $\varphi_*$. For $n>2$, the number of discontinuities of $\alpha_\varphi^{\prime\prime}$ increases (cf. \cite{KilianovaSevcovicANZIAM}).
\end{example}

In what follows, we shall denote by $\partial_x\alpha$ the total differential of the function $\alpha(x,\tau,\varphi)$ where $\varphi=\varphi(x,\tau)$, that is
\[
\partial_x\alpha (x,\tau,\varphi) = \alpha^\prime_x(x,\tau,\varphi) + \alpha^\prime_\varphi(x,\tau,\varphi)\, \partial_x \varphi,
\]
where $\alpha^\prime_x$ and $\alpha^\prime_\varphi$ are partial derivatives of $\alpha$ with respect to variables $x$ and $\varphi$, respectively.

The relationship between the transformed function $\varphi$ and the value function $V$ is given by the following theorem. 

\begin{theorem}\label{th-equiv-cons}
Assume that the utility function $u(x)$ and the intertemporal utility function $c(x,t)$ are $C^2$ smooth functions and such that $u$ is increasing and $c$ is non-decreasing in the $x$ variable. Then an increasing function $V(x,t)$ in the $x$ variable is a solution to the Hamilton-Jacobi-Bellman equation (\ref{eq_HJB}) if and only if the transformed function $\varphi(x,\tau) = -\partial_x^2 V(x,t)/\partial_x V(x,t),\,\ t=T-\tau$,  is a solution to the quasi-linear parabolic non-local PDE:
\begin{eqnarray}
&&-\partial_\tau \varphi + \partial_x\left(\partial_x\alpha(\cdot,\varphi) - \alpha(\cdot,\varphi)\varphi\right) 
=\frac{1}{b(T-\tau)} \partial_x\left( e^{\int_{x_*}^x\varphi(\xi,\tau)d\xi} \partial_x c\right),
\label{eq_PDEphi-cons}
\\
&&\varphi(x,0) = -u''(x)/u'(x),\quad (x,\tau)\in\R\times(0,T), 
 \label{init_PDEphi-cons}
\end{eqnarray}
and
\begin{equation}
V(x,t) = a(t) + b(t) \int_{x_*}^x e^{-\int_{x_*}^\xi \varphi(\eta,\tau) d\eta} d\xi, \quad t=T-\tau,
\label{Vexpression-cons}
\end{equation}
where the functions $a(t)$ and $b(t)$ are solutions to the system of ODEs:
\begin{eqnarray}
\frac{d}{dt} a(t) &=& \gamma(t) b(t) -c(x_*,t), \quad a(T)=u(x_*),
\label{ODEsystem-cons-a}
\\
\frac{d}{dt} b(t) &=& \omega(t) b(t) -\partial_x c(x_*,t), \quad b(T)=u^\prime(x_*).
\label{ODEsystem-cons-b}
\end{eqnarray}
Here $x_*\in\R$ is a fixed real number, $\gamma(t) := \alpha( x_*, \tau, \varphi(x_*,\tau))$, and $\omega(t) := \partial_x\alpha( x_*, \tau, \varphi(x_*,\tau)) - \alpha( x_*, \tau, \varphi(x_*,\tau)) \varphi(x_*,\tau)$ where $\tau=T-t$.
\end{theorem}

\noindent P r o o f. Let $V$ be a solution to the HJB equation (\ref{eq_HJB}) satisfying the terminal condition $V(x,T)=u(x)$ and such that $\partial_x V(x,t)>0$ for each $(x,t)\in\R\times[0,T)$. Thus $V$ solves (\ref{eq_HJBtransf}), i.e. $\partial_t V = \alpha(x,\tau,\varphi)\, \partial_x V - c$ where $\varphi=-\partial^2_x V/\partial_x V$. Therefore, $V$ is given by (\ref{Vexpression-cons}) with $a(t)=V(x_*,t)$ and $b(t)=\partial_x V(x_*,t)$. 

Since
\[
-\partial_\tau\varphi = -\frac{\partial^2_x\partial_t V}{\partial_x V} + \frac{\partial^2_x V \partial_x\partial_t V}{(\partial_x V)^2}
= -\frac{\partial^2_x\partial_t V}{\partial_x V} - \varphi\frac{\partial_x\partial_t V}{\partial_x V}, 
\]
\[
\partial^2_x V = -\varphi \partial_x V,\quad\hbox{and}\quad 
\partial^3_x V = - \partial_x(\varphi \partial_x V) = (\varphi^2 - \partial_x\varphi) \partial_x V,
\]
it follows from the equation $\partial_t V - \alpha(\cdot,\varphi) \partial_x V  + c =0$ that $\varphi$ satisfies:
\begin{eqnarray*}
-\partial_\tau\varphi &=& -\frac{1}{\partial_x V} \biggl( 
\partial^2_x\alpha\, \partial_x V + 2 \partial_x\alpha \, \partial^2_x V +\alpha\,\partial^3_x V + \varphi\partial_x\alpha\,\partial_x V +\varphi\alpha\,\partial^2_x V
 - \partial^2_x c -\varphi\partial_x c\biggr)
\\
&=& -\frac{1}{\partial_x V} \biggl( 
\partial^2_x\alpha\, \partial_x V -\varphi \partial_x\alpha\,\partial_x V +\alpha(\varphi^2-\partial_x\varphi)\partial_x V - \varphi^2\alpha\,\partial_x V
 - \partial^2_x c -\varphi\partial_x c
\biggr)
\\
&=& - \partial_x\left(\partial_x\alpha - \alpha\varphi\right)
+\frac{1}{\partial_x V(x,t)} 
\biggl( \varphi \partial_x c +\partial^2_x c\biggr)
\\
&=& - \partial_x\left(\partial_x\alpha - \alpha\varphi\right)
+\frac{e^{\int_{x_*}^x \varphi(\eta,t)d\eta}}{b(t)} 
\biggl( \varphi \partial_x c +\partial^2_x c\biggr)
\\
&=& - \partial_x\left(\partial_x\alpha - \alpha\varphi\right)
+ \frac{1}{b(t)} \partial_x \biggl(e^{\int_{x_*}^x \varphi(\eta,\tau)d\eta} \partial_x c\biggr), \qquad t=T-\tau.
\end{eqnarray*}
It means that the function $\varphi$ is a solution to the Cauchy problem 
(\ref{eq_PDEphi-cons})--(\ref{init_PDEphi-cons}). By differentiating (\ref{eq_HJBtransf}) with respect to $x$ we obtain $\partial_t \partial_x V = \partial_x(\alpha \partial_x V) -\partial_x c = \partial_x\alpha\, \partial_x V + \alpha \partial^2_x V -\partial_x c =  (\partial_x\alpha -  \alpha\,\varphi) \partial_x V -\partial_x c$. Taking $x=x_*$ we conclude 
$\partial_t \partial_x V(x_*,t) = \omega(t) \partial_x V(x_*,t) -\partial_x c(x_*,t)$. As $\partial_x V(x_*,T) =U^\prime(x_*)$ we obtain $b(t)=\partial_x V(x_*,t)$ is the solution to the ODE (\ref{ODEsystem-cons-b}). Furthermore, as $\partial_t V(x_*,t) = \alpha(x_*,\tau,\varphi(x_*,\tau)) \partial_x V(x_*,t) - c(x_*,t) = \gamma(t) b(t) - c(x_*,t)$, $\tau=T-t$, and $V(x_*,T)= u(x_*)$ we conclude  $a(t)=V(x_*,t)$ solves (\ref{ODEsystem-cons-a}), as claimed. 

On the other hand, suppose that a function $\varphi$ solves (\ref{eq_PDEphi-cons})--(\ref{init_PDEphi-cons}) and functions $a,b$ solve (\ref{ODEsystem-cons-a})--(\ref{ODEsystem-cons-b}). Then $V(x,t)$ given by (\ref{Vexpression-cons}) satisfies $-\partial^2_x V(x,t)/\partial_x V(x,t) = \varphi(x,\tau)$, $\tau=T-t$,  and \begin{eqnarray*}
V(x,T) &=& a(T) + b(T) \int_{x_*}^x e^{-\int_{x_*}^\xi \varphi(\eta,0) d\eta} d\xi 
\\
&=& u(x_*) + u'(x_*) \int_{x_*}^x e^{\int_{x_*}^\xi u''(\eta)/u'(\eta) d\eta} d\xi = u(x). 
\end{eqnarray*}

The function $b(t)$ satisfying (\ref{ODEsystem-cons-b}) is a positive function. Indeed, as $c$ is non-decreasing at $x_*$ we have
\[
\frac{d}{dt} \left(b(t) e^{-\int_t^T\omega(\eta)d\eta}\right) = -\partial_x c(x_*,t) e^{-\int_t^T\omega(\eta)d\eta} \le 0.
\]
Integrating the above inequality over $(t,T)$ we obtain 
$b(T) - b(t) e^{-\int_t^T\omega(\eta)d\eta} \le 0$ and so
\[
b(t)\ge b(T) e^{\int_t^T\omega(\eta)d\eta} = u'(x_*) e^{\int_t^T\omega(\eta)d\eta} >0,
\]
for any $t\in[0,T]$. Furthermore, as $\partial_x V(x,t) = b(t) e^{-\int_{x_*}^x \varphi(\xi,\tau) d\xi} >0$, the function $V(x,t)$ is increasing in the $x$ variable.

Note that for any $\xi$ we have
\begin{eqnarray*}
\int_{x_*}^x \partial_\xi V  \left(\partial_\xi\alpha -\alpha\,\varphi \right)d\xi
&=& \alpha\partial_x V - \gamma(t) b(t)  + \int_{x_*}^x -\partial^2_\xi V \, \alpha - \partial_\xi V\, \alpha \varphi d\xi
\\
&=&  \alpha\partial_x V - \gamma(t) b(t).
\end{eqnarray*}
Moreover, as $\varphi$ solves (\ref{eq_PDEphi-cons}), we have
\[
-\int_{x_*}^\xi \partial_\tau\varphi(\eta,\tau) d\eta = - \partial_\xi\alpha + \alpha\varphi +  \omega(t)
+ \frac{1}{b(t)}\left( e^{\int_{x_*}^\xi \varphi(\eta,\tau) d\eta} \partial_\xi c(\xi,t) - \partial_x c(x_*,t)  \right), \ \ t=T-\tau. 
\]
Differentiating (\ref{Vexpression-cons}) with respect to $t$ we obtain 
\begin{eqnarray*}
\partial_t V(x,t) &=& \frac{d a}{dt} +  \int_{x_*}^x e^{-\int_{x_*}^\xi \varphi(\eta,\tau) d\eta} \left(\frac{db}{dt} + b \int_{x_*}^\xi \partial_\tau\varphi(\eta,\tau) d\eta \right) d\xi 
\\
&=& \frac{d a}{dt} + \int_{x_*}^x e^{-\int_{x_*}^\xi \varphi(\eta,\tau) d\eta}\biggl(
\frac{db}{dt} + b ( \partial_\xi\alpha - \alpha\varphi) - b \omega 
\\
&& 
-e^{\int_{x_*}^\xi \varphi(\eta,\tau) d\eta} \partial_\xi c(\xi,t) + \partial_x c(x_*,t)
\biggr)d\xi
\\
&=&
 \frac{d a}{dt} + \int_{x_*}^x \partial_\xi V  \left(\partial_\xi\alpha -\alpha\,\varphi \right)d\xi - c(x,t) + c(x_*,t) 
 \\
 &=& \alpha(x,\tau,\varphi(x,\tau)) \partial_x V(x,t)  - c(x,t), \quad\tau=T-t, 
\end{eqnarray*}
which means that $V(x,t)$ solves equation (\ref{eq_HJBtransf}). Since $\partial_x V>0$, the function $V$ solves the HJB equation (\ref{eq_HJB}), as claimed.  
\hfill$\diamondsuit$

\medskip
Notice that the system of parabolic-ordinary differential equations (\ref{eq_PDEphi-cons})--(\ref{ODEsystem-cons-b}) can also be rewritten as a system of two quasi-linear parabolic equations. Indeed, let us denote 
\begin{equation}
\psi(x,\tau)=\frac{1}{b(t)}  e^{\int_{x_*}^x\varphi(\xi,\tau)d\xi}, \quad t=T-\tau.
\label{psi-fun}
\end{equation}
Then, by (\ref{Vexpression-cons}) we have $\psi(x,\tau)=1/\partial_x V(x,t)$. With regard to (\ref{eq_HJBtransf}) we obtain
\begin{eqnarray*}
- \partial_\tau \psi 
&=& -\frac{1}{(\partial_x V)^2} \partial_x\partial_t V
= -\frac{1}{(\partial_x V)^2}
 \left(
\partial_x\alpha\,\partial_x V + \alpha\, \partial_x^2 V -\partial_x c
\right)
\\
&=& -\partial_x\alpha\, \psi +\alpha\,\varphi\,\psi + \psi^2 \partial_x c. 
\end{eqnarray*}
As $\partial_x\psi=\varphi\,\psi$, we obtain  $\partial_x\varphi\, \psi = \partial_x^2\psi - \varphi\,\partial_x\psi$. Since $\partial_x\alpha = \alpha^\prime_\varphi\, \partial_x\varphi + \alpha^\prime_x$, we conclude that the function $\psi$ satisfies the following parabolic equation:
\[
-\partial_\tau\psi + \alpha^\prime_\varphi \partial^2_x\psi - (\alpha^\prime_\varphi \varphi +\alpha) \partial_x \psi + \alpha^\prime_x \psi - \psi^2 \partial_x c.
\]
The terminal condition $\psi(x,T)$ can be deduced from the terminal condition $V(x,T)=u(x)$. That is, $\psi$ satisfies the initial condition $\psi(x,0)=1/u'(x)$. In summary, we have shown the following theorem:

\begin{theorem}\label{th-equiv-phi-psi}
Assume that the terminal utility function $u(x)$ and the intertemporal function $c(x,t)$ are $C^2$ smooth functions and such that $u$ is increasing and $c$ is non-decreasing in the $x$ variable. Then an increasing function $V(x,t)$ in the $x$ variable is a solution to the Hamilton-Jacobi-Bellman equation (\ref{eq_HJB}) if and only if the pair $(\varphi,\psi)$ of transformed functions $\varphi(x,\tau) = -\partial_x^2 V(x,t)/\partial_x V(x,t)$ and $\psi(x,\tau)=1/\partial_x V(x,t),  t=T-\tau$, is a solution to the system of quasi-linear parabolic PDEs:
\begin{eqnarray}
&&-\partial_\tau \varphi + \partial_x\left(\partial_x\alpha(\cdot,\varphi) - \alpha(\cdot,\varphi)\varphi\right) 
=\partial_x\left(\psi \partial_x c\right),
\label{eq_PDEphi-cons-phi}
\\
&&-\partial_\tau\psi + \alpha^\prime_\varphi \partial^2_x\psi - (\alpha^\prime_\varphi \varphi +\alpha) \partial_x \psi + \alpha^\prime_x \psi - \psi^2 \partial_x c =0,
\label{eq_PDEphi-cons-psi}
\\
&&\varphi(x,0) = -u''(x)/u'(x),\quad \psi(x,0)=1/u'(x), \quad (x,\tau)\in\R\times(0,T), 
\label{init_PDEphi-cons-phi-psi}
\end{eqnarray}
and the value function $V(x,t)$ is given by (\ref{Vexpression-cons}).
\end{theorem}

At this point, we would like to emphasize the advantage of the suggested approach. By defining $\alpha$ as in \eqref{eq_alpha_def} and subsequently setting up the PDEs in \eqref{eq_PDEphi-cons} or \eqref{eq_PDEphi-cons-phi}--\eqref{eq_PDEphi-cons-psi}, one can compute the function $\alpha$ beforehand and then plug it into the corresponding PDEs as a known function. In this way, we do not have to deal with the maximization operator from \eqref{eq_HJB} in each $x$ and $t$ separately which significantly simplifies the computation process. 

\medskip
Next we derive a-priori bounds on a solution $\varphi(x,\tau)$ of (\ref{eq_PDEphi-cons}). We will use parabolic comparison principle (cf. \cite{Protter}). To this end, we need to restrict the form of the value function $\alpha$ and utility functions $u,c$ by the following assumptions:

\medskip

\begin{enumerate}
    \item[A1.] The value function $\alpha(x,\tau,\varphi)$ is separable in the following  sense:
    \[
    \alpha(x,\tau,\varphi)=\tilde\alpha(\varphi) + \alpha_0(x,\tau),
    \]
    where $\tilde\alpha$ is a $C^{1,1}$ smooth strictly increasing function with a bounded and Lipschitz continuous derivative for $\varphi\in(\varphi_{min},\infty)$ and $\alpha_0$ is a $C^2$ smooth function in $x\in\R$ and $\tau\in[0,T]$ variables.
    
    \item[A2.] There exist constants $\underline{\varphi}, \overline{\varphi}\in\R$ such that $\varphi_{min}\le\underline{\varphi}\le 0 \le \overline{\varphi}$, and the functions $\alpha_0$ and $c$ satisfy the estimates:
    \[
    \partial_x^2\alpha_0(x,\tau) -\overline{\varphi}\partial_x\alpha_0(x,\tau)\le 0 \ \le\  
        \partial_x^2\alpha_0(x,\tau) -\underline{\varphi}\partial_x\alpha_0(x,\tau),
    \]
    \[
    \overline{\varphi}\partial_x c(x,t) \ge - \partial_x^2 c(x,t) \ge \underline{\varphi}\partial_x c(x,t), \quad \partial_x c(x,t)\ge 0, \ t=T-\tau,
    \]
    for any $x\in\R$ and $\tau\in[0,T]$.
\end{enumerate}

\begin{example}
If $\alpha(x,\tau,\varphi) =\tilde\alpha(\varphi) - \varepsilon(\tau) e^{-x} -r(\tau)$ with $\varepsilon\ge0$ is the value function introduced in Section~2 then $\tilde\alpha$ is defined on $[-1,\infty)$ and $\alpha_0(x,\tau) =  - \varepsilon(\tau) e^{-x} -r(\tau)$ satisfies the assumption (A2) with $\underline{\varphi}=-1$ and any $\overline{\varphi}\ge 0$.

The intertemporal utility function $c$ of the form $c(x,t)=-\kappa e^{-d x -\varrho(T-t)}$ with $\kappa,d \ge 0$ satisfies (A2) if $\overline{\varphi}\ge d$.

\end{example}

\begin{theorem}\label{th-bounds}
Assume that the utility function $u(x)$ is a $C^2$ smooth strictly increasing function for $x\in\R$. Assume that the value function $\alpha$ and intertemporal utility function $c$ satisfy Assumption (A) with constants $\underline{\varphi}\le 0 \le \overline{\varphi}$.

If the utility function $u$ satisfies the inequalities
$\underline{\varphi} \le \varphi(x,0) = -u''(x)/u'(x) \le \overline{\varphi}$ for any $x\in\R$, then, for the bounded solution $\varphi$ to (\ref{eq_PDEphi-cons}) we have a-priori estimate: $\underline{\varphi} \le \varphi(x,\tau) \le \overline{\varphi}$ for any $\tau\in[0,T)$ and $x\in\R$.

\end{theorem}

\noindent P r o o f. Let $\psi(x,\tau)$ be a $C^2$ smooth nonnegative function, $\psi(x,\tau)\ge 0$. Let us define the parabolic operator:
\[
{\mathcal L}(\varphi)\equiv
-\partial_\tau \varphi + \partial_x\left(\partial_x\alpha(\cdot,\varphi) - \alpha(\cdot,\varphi)\varphi\right) - \varphi\,\psi\,\partial_x c\,.
\]
Since $\alpha(x,\tau,\varphi)=\tilde\alpha(\varphi) + \alpha_0(x,\tau)$, we have ${\mathcal L}(\tilde\varphi) = 
\partial_x^2 \alpha_0 - \tilde\varphi\partial_x\alpha_0 - \tilde\varphi\,\psi\,\partial_x c$ for any constant function $\tilde\varphi\in\R$. Thus, for  constant functions $\underline{\varphi}, \overline{\varphi}$ and nondecreasing function $c$, and  $\alpha_0$ satisfying assumption (A2) we have  
\[
{\mathcal L}(\underline{\varphi}) \ge -\underline{\varphi} \psi \partial_x c \ge  \psi \partial^2_x c \ge -\overline{\varphi} \psi \partial_x c \ge {\mathcal L}(\overline{\varphi}).
\]
Now let $\varphi$ be a solution to (\ref{eq_PDEphi-cons}). Then ${\mathcal L}(\varphi) = \psi \partial^2_x c$ where 
$\psi(x,\tau)=\frac{e^{\int_{x_*}^x\varphi(\xi,\tau)d\xi}}{b(T-\tau)} > 0$. Note the fact that $\partial_x\psi=\varphi\,\psi$. Hence the bounded solution $\varphi$ satisfies the following inequalities: 
${\mathcal L}(\underline{\varphi}) \ge {\mathcal L}(\varphi) \ge {\mathcal L}(\overline{\varphi})$. 

If the initial condition satisfies the inequalities
$\underline{\varphi} \le \varphi(x,0) = -u''(x)/u'(x) \le \overline{\varphi}$ for any $x\in\R$ then, applying the parabolic comparison principle (cf. \cite{Protter}), we obtain that the bounded solution $\varphi$ to (\ref{eq_PDEphi-cons}) satisfies the inequalities $\underline{\varphi} \le \varphi(x,\tau) \le \overline{\varphi}$ for any $\tau\in[0,T]$ and $x\in\R$, as claimed. \hfill$\diamondsuit$

By $H^{k+\lambda}(\Omega)$, $0<\lambda<1$, we denote the Banach space consisting of all uniformly continuous functions $\varphi$ on $\bar\Omega=[x_L, x_R]$ whose $k$-th derivative is uniformly $\lambda$-H\"older continuous, i.e. the H\"older semi-norm $\langle \varphi \rangle^{(\lambda)}_k =
\sup_{x,y\in\Omega, x\not= y} |\partial^k_x\varphi(x) - \partial^k_x\varphi(y)|/|x-y|^\lambda$ is finite. Let $Q_T=\Omega \times (0,T)$ be a bounded cylinder. Following Ladyzhenskaya \emph{et al.} \cite{LSU} we introduce the parabolic H\"older space  $H^{2k+\lambda, k+\lambda/2}(Q_T)$ consisting of all continuous functions $\varphi:Q_T\to\R$ such that functions $\partial^k_\tau\varphi, \partial^{2k}_x\varphi$ are $\lambda$-H\"older continuous in the $x$-variable and $\lambda/2$-H\"older continuous in the $\tau$-variable.

\begin{theorem}\label{existence}
Let $\Omega=(x_L,x_R)$ be a bounded interval. Assume $\alpha(x,\tau,\varphi)$ is $C^2$ smooth in the $t\in[0,T]$ and $x\in\Omega$ variables,  $C^{1,1}$ smooth in the $\varphi$ variable, and such that $0<\alpha^\prime_- \le \alpha^\prime_\varphi(x,\tau,\varphi) \le \alpha^\prime_+ <\infty$ for any $x\in\Omega,\tau\in[0,T]$, and $\varphi\ge \varphi_{min}$. Assume  $c\in H^{2+\lambda, 1+\lambda/2}(\Omega_T)$ for some $0<\lambda<1/2$, and $c(x,t)$ is a non-decreasing function in the $x$ variable.

If the initial conditions $\varphi(\cdot,0), \psi(\cdot,0)\in H^{2+\lambda}(\Omega)$, then there exists a classical solution $(\varphi, \psi)$ to the system of quasi-linear parabolic equations (\ref{eq_PDEphi-cons-phi})--(\ref{init_PDEphi-cons-phi-psi}) satisfying prescribed Dirichlet boundary conditions at $x_L, x_R$. Moreover, $\eta,\psi\in H^{2+\lambda, 1+\lambda/2}(\Omega_T)$ where $\eta(x,\tau)=\alpha(x,\tau,\varphi(x,\tau))$. The function $\tau\mapsto \partial_\tau\varphi(x,\tau)$ is $\lambda/2$-H\"older continuous for all $x\in\R$ whereas $x\mapsto\partial_x\varphi(x,\tau)$ is Lipschitz continuous for all $\tau\in[0,T]$.

\end{theorem}

\begin{proof}
The proof is analogous to that of \cite[Theorem 5.3]{KilianovaSevcovicANZIAM} where we proved existence of H\"older smooth solutions in the case when $c\equiv 0$. The methodology of the proof is based on the Schauder type of estimates (c.f. \cite{LSU}). 

Since the diffusion function $\alpha$ need not be $C^{2+\lambda}$ smooth in the $\varphi$ variable we first rewrite the system (\ref{eq_PDEphi-cons-phi})--(\ref{init_PDEphi-cons-phi-psi}) using the auxiliary function $\eta=\alpha(\cdot,\varphi)$. Then $\varphi=\beta(\cdot,\eta)$ where $\beta$ is the inverse function to the strictly increasing function $\alpha$, i.e. $\varphi=\beta(\cdot, \alpha(\cdot, \varphi))$. Moreover, $\beta^\prime_\eta = 1/\alpha^\prime_\varphi$ and $\partial_\tau\varphi = \beta^\prime_\eta \partial_\tau\eta + \beta^\prime_\tau$. Then system (\ref{eq_PDEphi-cons-phi})--(\ref{init_PDEphi-cons-phi-psi})  can be rewritten in the form:
\begin{equation}
\partial_\tau\Phi = \delta(\cdot, \eta) \partial^2_x\Phi + F(\cdot,\Phi),
\quad \Phi(x,0)=\Phi_0,
\label{Phi}
\end{equation}
where $\Phi=(\eta,\psi)$, $\delta(\cdot,\eta)= \alpha^\prime_\varphi(\cdot, \beta(\cdot, \eta))=1/\beta^\prime(\cdot,\eta)$, and
\[
F(\cdot,\Phi) = \bigl(-\delta [ \partial_x(\eta\beta) + \partial_x(\psi\partial_x c) + \beta^\prime_\tau],\ -(\beta/\beta^\prime_\eta +\eta)\partial_x\psi +\alpha^\prime_x \psi -\psi^2\partial_x c\bigr).
\]
Note that there are constants $\alpha^\prime_\pm$ such that $0<\alpha^\prime_- \le \delta \le \alpha^\prime_+ <\infty$. Applying a  $C^{2+\lambda}$ regularization of the function $\eta\mapsto\delta(\cdot,\eta)$ and following the proof of \cite[Theorem 5.3]{KilianovaSevcovicANZIAM} and the result on existence of classical solutions to the regularized equation (c.f. \cite[Ch. V, pp. 495-496]{LSU}) we conclude existence of a weak solution $\Phi\in W^{2,1}_2(Q_T)$ of (\ref{Phi}) satisfying the prescribed Dirichlet boundary and initial conditions.  Recall that the parabolic Sobolev space $W^{2,1}_2(Q_T)$ consists of all square integrable functions $\Phi\in L_2(Q_T)$ such that  $\partial_x\Phi, \partial^2_x\Phi, \partial_\tau\Phi \in L_2(Q_T)$. The space $W^{2,1}_2(Q_T)$ is continuously embedded into the H\"older space $H^{\lambda, \lambda/2}(Q_T)$ for $0<\lambda<1/2$ (c.f. \cite{LSU}).
The rest of the proof is based on a simple bootstrap argument. A weak solution $\Phi\in W^{2,1}_2(Q_T)$ is a solution to the linear equation: 
\begin{equation}
\partial_\tau\Phi = \tilde\delta(x,\tau) \partial^2_x\Phi + \tilde B_1(x,\tau) \partial_x \Phi + \tilde B_0(x,\tau) \Phi,
\quad \Phi(x,0)=\Phi_0(x), 
\label{PhiLin}
\end{equation}
where the diffusion coefficient  $\tilde\delta(\cdot) = \delta(.,\eta(\cdot))$, and $2\times2$ matrices $\tilde B_1, \tilde B_0$ belong to $H^{\lambda, \lambda/2}(Q_T)$ because $\eta,\psi \in H^{\lambda, \lambda/2}(Q_T)$ and $\delta$ is a Lipschitz continuous function.  According to \cite[Theorem 12.2, Chapter III]{LSU} we have $(\eta,\psi)\equiv\Phi\in H^{2+\lambda, 1+\lambda/2}(Q_T)$ where $\eta=\alpha(\cdot,\phi)$. The proof of theorem now follows.
\end{proof}

\section{Numerical approximation scheme}
\label{sec:numscheme}

A numerical scheme that we propose for solving quasi-linear parabolic equation (\ref{eq_PDEphi-cons}) is based on a semi-implicit in time approximation method. Spatial discretization is based on a finite volume approximation scheme (cf. LeVeque \cite{LeV}) combined with a nonlinear equation iterative solver method  proposed by Mikula and K\'utik in \cite{KutikMikula}. This methodology for solving the transformed HJB equation was proposed and analyzed in \cite{KilianovaSevcovicANZIAM} and \cite{KilianovaSevcovicKybernetika} for the case of absence of an  intertemporal utility function, i.e. $c=0$. In such a case, analysis of the experimental order of convergence suggested  the second order of convergence with respect to the spatial discretization step (see \cite{KutikMikula}, \cite{KilianovaSevcovicANZIAM}). 

Equation (\ref{eq_PDEphi-cons}) belongs to a wide class of quasi-linear parabolic equations of the general form:
\begin{equation}
\partial_\tau \varphi =  \partial_{x}^2 A(x,\tau,\varphi) + \partial_{x} B(x,\tau,\varphi)+ C(x,\tau,\varphi), 
\label{eq_numerics_general}
\end{equation}
satisfying the initial condition $\varphi(x,0)= -u''(x)/u'(x)$ where  $x\in\R, \tau\in(0,T)$. Here $A(x,\tau,\varphi)=\alpha(x,\tau,\varphi),\quad  B(x,\tau,\varphi)= - \alpha(x,\tau,\varphi)\varphi$,  and 
\[
C(x,\tau,\varphi) = -\frac{e^{\int_{x_*}^x \varphi(\eta,\tau)d\eta}}{b(T-\tau)} 
\biggl( \varphi(x,\tau) \partial_x c(x,T-\tau) +\partial^2_x c(x,T-\tau)\biggr).
\]

\subsection{A semi implicit time-space discretization of the transformed non-local parabolic equation}

Since the original spatial domain for the $x$ variable is unbounded, we first truncate it into  a bounded computational domain $[x_L, x_R]$ and we use uniform spatial discretization mesh points $x_i = x_L+ i h$ for
$i=0,\cdots,n+1$ where $h=(x_R-x_L)/(n+1)$. Thus $x_0 =x_L$ and $x_{n+1}=x_R$. Following the idea of dual finite volumes (cf. \cite{LeV}), the inner mesh points $x_i$, $i=1,\cdots,n$, are the centers of dual finite volumes $(x_{i - \frac{1}{2}}, x_{i+\frac{1}{2}})$. In what follows, the dual volumes will be denoted by $(x_{i-}, x_{i+})$, i.e. $x_{i\pm} = x_{i\pm\frac{1}{2}}$. Clearly, $h=x_{i+} - x_{i-}$. The time discretization levels are set to $\tau^j=j k, j=0, \cdots, m$, where $k=T/m$ and $m$ is the number of time discretization  steps. If we integrate equation (\ref{eq_PDEphi-cons}) over dual finite volumes, apply the midpoint rule on the left-hand side integral and approximate the time derivative by the Euler forward finite difference, we arrive at the following system of equations:
\begin{equation}
\varphi_i^{j+1}  = \frac{k}{h} (I_i + J_i) + \varphi_i^j\,, \quad
i=1, \cdots ,n,\ j=0, \cdots, m-1,
\end{equation}
where 
\begin{equation}
I_i =\int_{x_{i-}}^{x_{i+}} \partial_x \left(\partial_x A + B\right) {\d}x
=  \left[A^\prime_x + A^\prime_{\varphi} \partial_x\varphi + B\right]_{x=x_{i-}}^{x=x_{i+}},
\label{intI}
\end{equation}
and the integral $\int_{x_{i-}}^{x_{i+}} C  {\d}x$ over the interval $(x_{i-}, x_{i+})$ is approximated by means of the mid-point rule integration, that is 
\begin{equation}
J_i = \int_{x_{i-}}^{x_{i+}} C  {\d}x
\approx 
-h \frac{e^{\int_{x_*}^{x_i} \varphi(\eta,\tau)d\eta}}{b(T-\tau)} 
\biggl( \varphi_i \partial_x c(x_i,T-\tau) +\partial^2_x c(x_i, T-\tau)\biggr)
\,.
\label{intJ}
\end{equation}

Let us denote 
\[
D_{i\pm}^j = A^\prime_{\varphi}(x_{i\pm}, \tau^j, \varphi_{i\pm}^j), 
\quad 
E_{i\pm}^j = A^\prime_x(x_{i\pm},\tau^j,\varphi_{i\pm}^j), 
\quad
F_{i\pm}^j = B(x_{i\pm},\tau^j,\varphi_{i\pm}^j),
\]
and approximate the derivatives $\partial_x\varphi$ at dual mesh points $x_{i\pm}$ by the central differences:
\[
\partial _x \varphi|_{i+}^j \approx \frac{\varphi_{i+1}^j-\varphi_i^j}{h},
\quad
\partial _x \varphi|_{i-}^j \approx \frac{\varphi_i^j-\varphi_{i-1}^j}{h}.
\]

Let us fix a point $x_*=x_{i_*}$ for some spatial index $i_*$. As for the integral  $\int_{x_*}^x \varphi(\eta,\tau)d\eta$ appearing in the non-local term $J_i$ at time layer $j$ (denoted as $J_i^j$) we use the trapezoidal integration rule:
\[
\int_{x_*}^{x_i} \varphi^j(\eta,\tau)d\eta \approx \Phi^j_i - \Phi^j_{i_*}
\]
where $\Phi^j_i = \frac{h}{2} (\varphi^j(x_L) + 2\varphi^j(x_1) + \dots + 2\varphi^j(x_{i-1}) + \varphi^j(x_i))$, which can be efficiently calculated recursively as follows:
\[
\Phi_1^j = \frac{h}{2} (\varphi^j(x_L) + \varphi^j(x_1)), \qquad \Phi^j_{i+1} = \Phi^j_i + \frac{h}{2}(\varphi^j(x_i) + \varphi^j(x_{i+1}))\quad \text{ for } i=1,\cdots,n-1.
\]
Hence
\begin{equation}
J_i^j = 
-h \frac{e^{\Phi^j_i - \Phi^j_{i_*}}}{b^j} 
\biggl( \varphi_i^j \partial_x c(x_i,T-\tau^j) +\partial^2_x c(x_i,T-\tau^j)\biggr).
\label{intJdiscr}
\end{equation}
Here  $b^j$ is a discrete explicit/implicit  Euler approximation of the solution $b(T-\tau^j)$ to the ODE (\ref{ODEsystem-cons-b}): $-db/d\tau = \omega b -\partial_x c(x_*,T-\tau)$, i.e. 
\begin{equation}
b^{j+1} = (1 -  k \omega^j) b^j + k \partial_x c(x_*, T-\tau^j), \quad b^0 = u'(x_*), \ j=0, \cdots, m-1, 
\label{bdiscr}
\end{equation}
when treated explicitly, or
\begin{equation}
b^j = \frac{1}{1+k \omega^j} (b^{j-1} + k \partial_x c(x_*, T-\tau^j)), \quad b^0 = u'(x_*), \ j=1, \cdots, m, 
\label{bdiscr_implicit}
\end{equation}
when treated implicitly. Here $\omega^j = \omega(T-\tau^j) =(\partial_x\alpha - \alpha\varphi)|_{x=x_*}$ can be approximated by 
\begin{eqnarray*}
\omega^j &=& \alpha^\prime_x(x_*, \tau^j, \varphi_{i_*}^j) 
+ \alpha^\prime_\varphi (x_*, \tau^j, \varphi_{i_*}^j) \frac{\varphi_{i_*+1}^j-\varphi_{i_*-1}^j}{2h}
 - \alpha(x_*, \tau^j, \varphi_{i_*}^j)\varphi_{i_*}^j
\\
&=& E^j_{i^*} + D^j_{i^*} \frac{\varphi_{i_*+1}^j-\varphi_{i_*-1}^j}{2h} + F^j_{i^*} . 
\end{eqnarray*}
To compute a solution at the new time layer $\tau^{j+1}$, we take the terms $D_{i\pm}^j, E_{i\pm}^j, F_{i\pm}^j$ from the previous time layer $\tau^j$ and the terms $\partial_x \varphi|_{i\pm}^{j+1}$ from the new layer $\tau^{j+1}$. Rearranging the new-layer terms to the left-hand side and the old-layer terms to the right-hand side, we obtain a tridiagonal system of linear algebraic equations:
\begin{eqnarray}
-\frac{k}{h^2}D_{i+}^j \varphi_{i+1}^{j+1}
&+& (1+\frac{k}{h^2}(D_{i+}^j+D_{i-}^j)) \varphi_i^{j+1} -
\frac{k}{h^2} D_{i-}^j \varphi_{i-1}^{j+1}
\nonumber\\
&=& \frac{k}{h}(J_i^j + E_{i+}^j- E_{i-}^j + F_{i+}^j - F_{i-}^j)
+ \varphi_i^j\,,
\label{discretescheme}
\end{eqnarray}
which can be efficiently and fastly solved by the Thomas algorithm.

We assume Neumann boundary conditions at the boundaries $x_L, x_R$. More precisely, $\partial_x \varphi(x,\tau) = 0 \text{ at } x=x_L, x_R$,
for all $\tau \in (0,T]$. The boundary conditions can be deduced from the asymptotic behavior of equation (\ref{eq_PDEphi-cons}) for $x\to\pm \infty$. After discretization, these boundary conditions attain the form:
$$
\varphi_0^{j} = \varphi_1^j, \qquad \varphi_{n+1}^j  = \varphi_n^j.
$$

\subsection{Comparison with policy iteration method for solving HJB equations}

In this section we discuss comparison of the numerical approximation scheme (\ref{discretescheme}) and the fixed policy iteration method for solving HJB equation investigated by Huang {\it et al.} \cite{Huang2010} and Reisinger and Witte \cite{Reisinger}.

Denote $V^j=V(\cdot, T-\tau^j), c^j=c(\cdot, T-\tau^j)$. Then the time implicit time discretization of the HJB equation (\ref{eq_HJB}) can be written as follows:
\begin{equation}
-\frac{V^j-V^{j-1}}{k}
+ \max_{\bmtheta\in\Delta} \left(\mu(\cdot,\bmtheta)\partial_x V^j
+\frac{1}{2}\sigma(\cdot,\bmtheta)^2 \partial^2_x V^j\right) + c^j =0, \quad V^0=u,
\label{discr-HJB}
\end{equation}
for $j=1,\cdots, m$. That is,
\begin{equation}
-\frac{V^j-V^{j-1}}{k}
- \left(- \mu(\cdot,\bmtheta^j)\partial_x V^j
+\frac{1}{2}\sigma(\cdot,\bmtheta^j)^2 \partial^2_x V^j\right) + c^j =0, 
\label{discr-HJB-A}
\end{equation}
where
\begin{equation}
\bmtheta^j= \arg\min_{\bmtheta\in\Delta} \left(-\mu(\cdot,\bmtheta)\partial_x V^j
- \frac{1}{2}\sigma(\cdot,\bmtheta)^2 \partial^2_x V^j\right).
\label{discr-theta}
\end{equation}
The fixed policy iteration method consists of replacing $\bmtheta^j$ by $\bmtheta^{j-1}$ in (\ref{discr-HJB-A}) and solving a linear equation for $V^j$, i.e. 
\[
-\frac{V^j-V^{j-1}}{k}
- \left(- \mu(\cdot,\bmtheta^{j-1})\partial_x V^j
-\frac{1}{2}\sigma(\cdot,\bmtheta^{j-1})^2 \partial^2_x V^j\right) + c^j =0. 
\]
Since 
\begin{eqnarray*}
- \mu(\cdot,\bmtheta^{j-1})\partial_x V^j
&-&\frac{1}{2}\sigma(\cdot,\bmtheta^{j-1})^2 \partial^2_x V^j
= (- \mu(\cdot,\bmtheta^{j-1})
 +\frac{1}{2}\sigma(\cdot,\bmtheta^{j-1})^2 \varphi^j  ) \partial_x V^j
\\
&&= 
(- \mu(\cdot,\bmtheta^{j-1})
 +\frac{1}{2}\sigma(\cdot,\bmtheta^{j-1})^2 \varphi^{j-1}  + 
\frac{1}{2}\sigma(\cdot,\bmtheta^{j-1})^2 (\varphi^j-\varphi^{j-1})
) \partial_x V^j
\\
&&=
(\alpha(\cdot, \varphi^{j-1}) + \alpha^\prime_\varphi(\cdot,\varphi^{j-1}) (\varphi^j-\varphi^{j-1}) ) \partial_x V^j,
\end{eqnarray*}
the fixed policy iteration method for solving HJB equation (\ref{eq_HJBtransf})
 corresponds to the numerical solution of the  transformed equation (\ref{eq_PDEphi-cons}) by means of the semi-implicit scheme (\ref{discretescheme})  in which $\alpha$ is approximated by its linearization at $\varphi^{j-1}$ from the previous time step $\tau^{j-1}$. 

\begin{remark}
The main advantage of our method is twofold. First, we work with a transformed function $\varphi$ representing risk aversion of the investor. The boundary conditions for a truncated domain can be set up in a natural way, e.g. homogeneous Neumann boundary conditions. For the original problem formulated in terms of the intertemporal value function $V$, one can expect unbounded exponential like solution and numerical problems when  working small and large values of $V$ and treatment of boundary conditions for $V$. 

Secondly, the advantage consists of the possibility of evaluation of the value function $\alpha$ in a fast and efficient way instead of computation of $\bmtheta^{j}$ in (\ref{discr-theta}). This might be useful when treating problems involving convex conic programming, e.g. worst-case portfolio selection problem, for which one can use efficient tools for solving  convex conic programming optimization problems (cf. \cite{KilianovaTrnovska}).

\end{remark}

\section{Computational results and conclusions}

\subsection{Numerical benchmark to a traveling wave solution}

Suppose that the value function $\alpha$ depends only on the $\varphi$ variable (e.g. we set $\varepsilon=0$, $r=0$ in \eqref{volatility}--\eqref{drift}) and the intertemporal utility function is given as follows: 
\[
c(x,t)= W(x-v(T-t)), \quad \hbox{where}\quad 
W(\xi) = \left(-v +\alpha(-u''(\xi)/u'(\xi))\right)\, u'(\xi).
\] 
Note that $c(x,T-\tau)= W(x-v\tau)$. Here $v\in\R$ is a given constant traveling wave  speed. Then the function $V(x,t)=u(x-v(T-t))$ satisfies the equation:
\begin{eqnarray*}
\partial_t V(x,t) &-& \alpha\left(-\partial^2_x V(x,t)/\partial_x V(x,t) \right) \partial_x V(x,t)
\\
&=& - \left( - v + \alpha(-u''(x+v(T-t))/u'(x-v(T-t)))\right) u'(x-v(T-t))
\\
&=& - W(x-v(T-t)) = - c(x,t),
\end{eqnarray*}
i.e. $V(x,t)$ is a solution to the HJB equation (\ref{eq_HJBtransf}) and $V(x,T)=u(x)$. Hence the function 
\begin{equation}
\varphi(x,\tau) = -u''(x-v\tau)/u'(x-v\tau)
\label{travelingwavesol}
\end{equation}
is a traveling wave solution to (\ref{eq_PDEphi-cons}) satisfying the initial condition 
$\varphi(x,0) = -u''(x)/u'(x)$. The explicit solution of the form (\ref{travelingwavesol}) can be used to test our numerical approximation scheme. As a testing example one can consider utility and value functions of the form:
\[
u(x) = \arctan(x), \quad \alpha(\varphi)=\varphi - 1/(\varphi+2).
\]
Then $u$ represents a convex-concave utility function with variable absolute risk aversion $a(x)$ given by
\[
a(x)= -\frac{u''(x)}{u'(x)} = \frac{2x}{1+x^2}.
\]
If we set $c(x,t)=  W(x-v(T-t))$ then $\varphi(x,\tau) = a(x-v\tau)=-u''(x-v\tau)/u'(x-v\tau)$ is a solution to (\ref{eq_PDEphi-cons}) satisfying the initial condition $\varphi(x,0)= a(x) = 2x/(1+x^2)$. Consequently, $V(x,t)= u(x-v (T-t))$ is the traveling wave solution to the HJB equation (\ref{eq_HJBtransf}). The traveling wave solution $\varphi$ is depicted in Fig.~\ref{fig:tw-solutions} (left) for times $\tau^j = j T/10, j=0,\cdots, 10$, where $T=1$ and $v=5$. As for a numerical solution, we considered the truncated computational domain $[x_L,x_R] = [-20,20]$ and Dirichlet boundary conditions $\varphi(x_L,\tau)=a(x_L-v\tau),\  \varphi(x_R,\tau)=a(x_R-v\tau)$, for all $\tau>0$, which coincide with exact values of the explicit solution.

\begin{figure}
    \centering
    \includegraphics[width=0.45\textwidth]{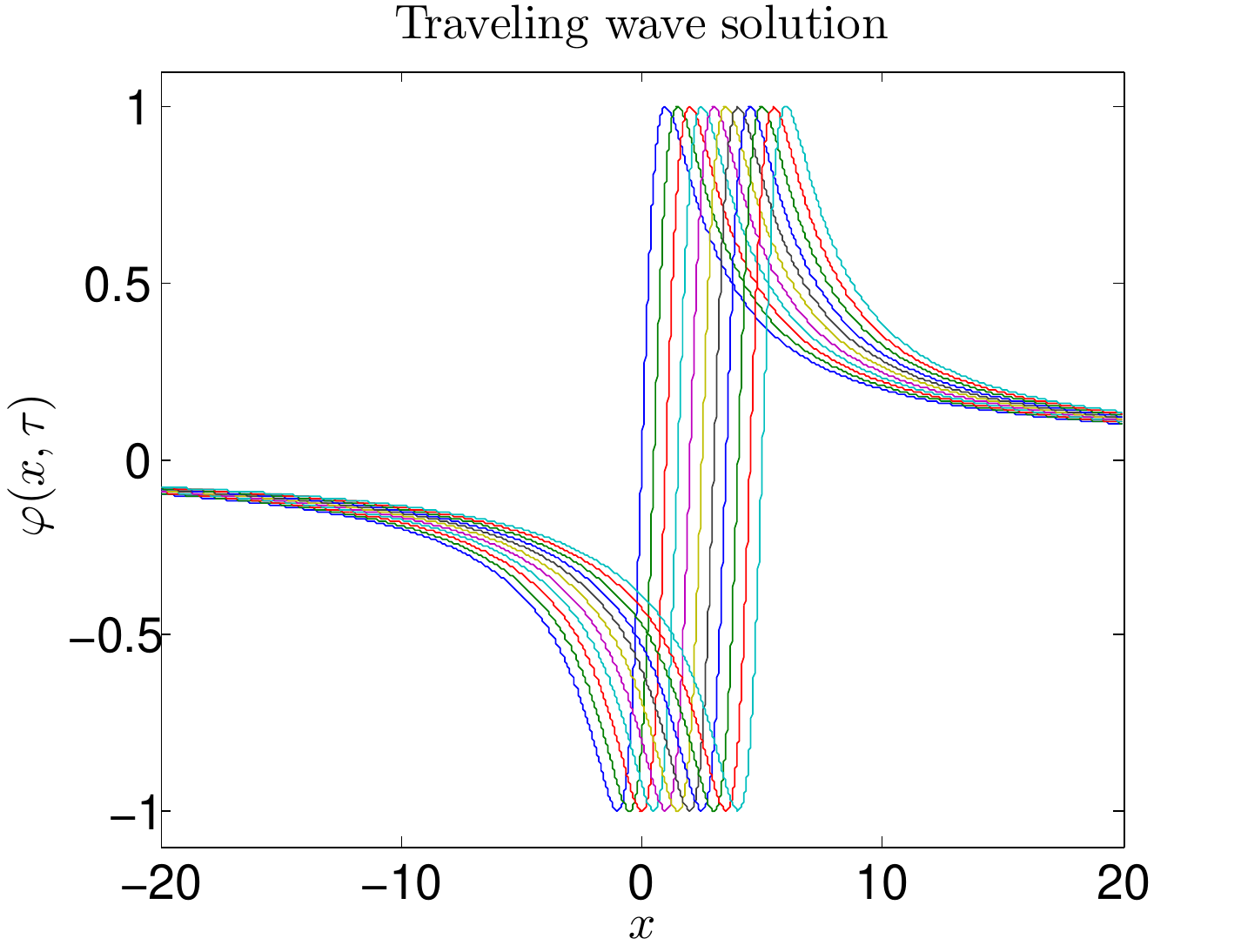}
    \includegraphics[width=0.43\textwidth]{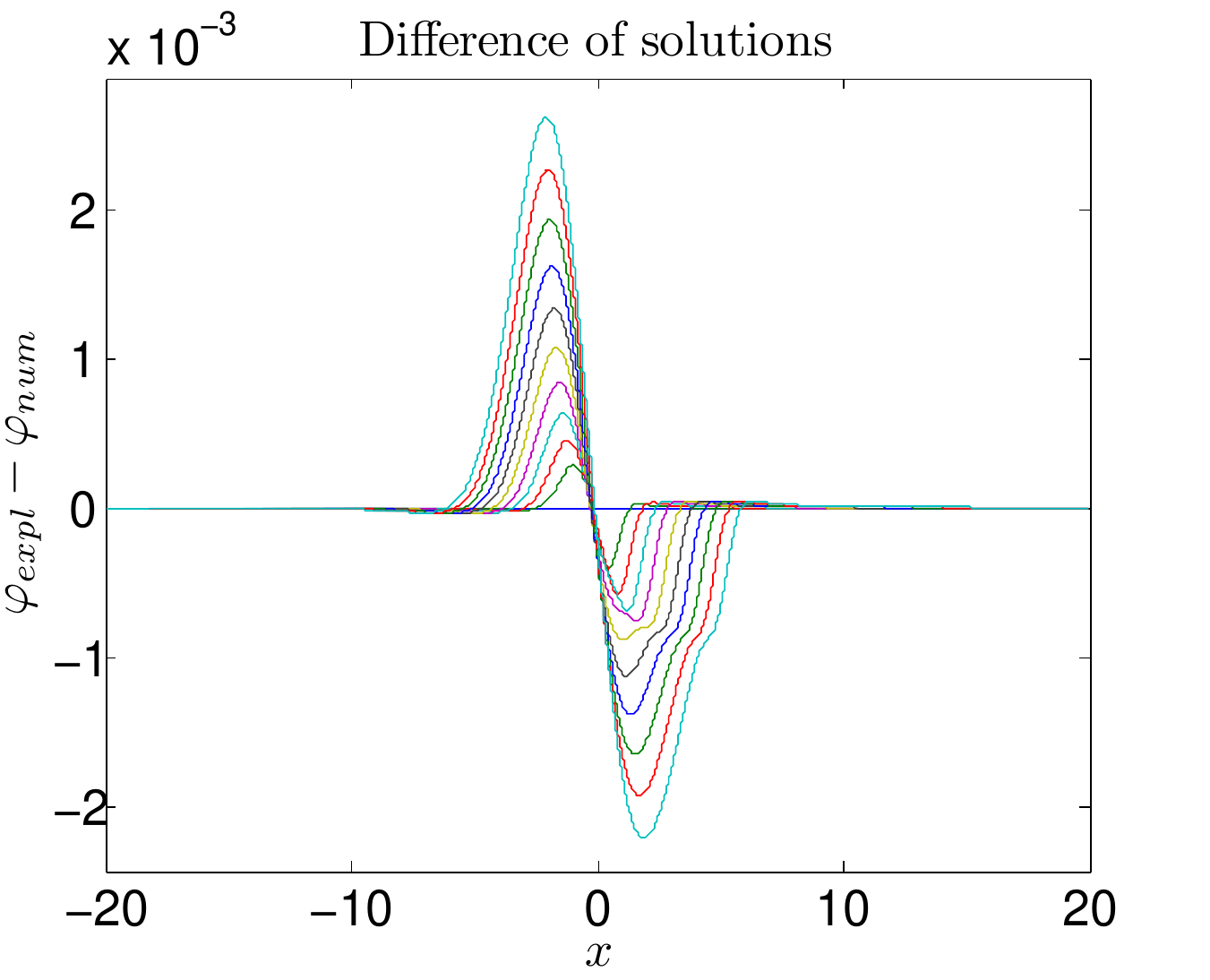}
    \caption{(Left) A graph of the traveling wave solution $\varphi(x,\tau)$ and (right) graph of the difference between the explicit and numerical solution for times $\tau=j T/10,\  j=0,\cdots,10$, and parameters $T=1, v=5, h=0.01$. }
    \label{fig:tw-solutions}
\end{figure}

Let $\varphi_{expl}$ be the explicit traveling wave solution and $\varphi_{num}$ be the numerical solution constructed by means of our approximation scheme presented in Section~5. The $L_2$ and $L_\infty$ discrete norms are defined as follows:
\[
\Vert\varphi\Vert_{L_2}= \sqrt{h\sum_i{\varphi_i^2}},
\quad 
\Vert\varphi\Vert_{L_\infty}= \max_i|\varphi_i|,
\]
and the error between solutions as
\[
error_{\infty,p}(h)=\Vert\varphi_{expl} - \varphi_{num}\Vert_{L_\infty((0,T):L_p)}=
\max_{\tau^j}\Vert\varphi_{expl}(\cdot,t^j) - \varphi_{num}(\cdot, t^j)\Vert_{L_p}, \quad p=2,\infty.
\]
We consider the following relation between spatial and time discretization steps: $k=h^2$.
Supposing $error(h)=O(h^\delta)$, estimation of the order parameter $\delta$ can be obtained by means of the experimental order of convergence (or convergence ratio). It can be defined with respect to the norm of the space $L_\infty((0,T): L_2(x_L,x_R))$ as follows:
\[
EOC_j = \frac
{\ln (error(h_{j+1})/error(h_j))}
{\ln (h_{j+1}/h_{j})}, \quad j=1, \cdots, J,
\]
where $h_1>h_2 > \cdots > h_J$.  The results of computation of EOCs are summarized in Table~\ref{tab:eoc}.  

The numerical results indicate the second order of convergence of the proposed numerical method. This is in accordance with the earlier result of the authors in \cite{KilianovaSevcovicANZIAM} where we showed the same order of experimental convergence  for the special case when  $c=0$.

An example of the difference of explicit and numerical solution  $\varphi_{expl} - \varphi_{num}$ is depicted in Fig.~\ref{fig:tw-solutions} (right). We can observe that the error is largest where the function $\varphi$ is steep.

\begin{table}
\caption{\label{tab:eoc} The $L_{\infty}((0,T):L_2(x_L,x_R))$ and $L_{\infty}((0,T):L_{\infty}(x_L,x_R))$ norm of the error of the numerical
solution with the spatial step $h$ and time step $k=h^2$ and the exact
traveling wave solution. Corresponding experimental orders of convergence.}

\begin{center}
\begin{tabular}{l||c|c||c|c}
   $h$  & $L_{\infty}((0,T):L_2)$-err & $EOC_{k=h^2}$  & $L_{\infty}((0,T):L_{\infty}))$-err & $EOC_{k=h^2}$\\ \hline \hline
   0.05 &  1.1886e-01  &  --  &  5.8577e-02 &  --  \\
   0.025 & 3.2102e-02 &  1.8885  &  1.5919e-02  &  1.8796  \\
   0.0125 & 8.1969e-03 &  1.9695  &  4.0718e-03 &  1.9670  \\
   0.01 & 5.2598e-03  &  1.9882   &   2.6133e-03 &   1.9874  \\
   0.005 &  1.3196e-03   &   1.9949   &   0.6558e-03   &    1.9945  \\
   \hline
\end{tabular}
\end{center}

\end{table}

\subsection{Dynamic portfolio optimization example}

Now we illustrate the solution of the proposed scheme on an example of dynamic portfolio optimization. Following Kilianov\'a and \v{S}ev\v{c}ovi\v{c} 
\cite{KilianovaSevcovicANZIAM, KilianovaSevcovicKybernetika}, we consider a stochastic dynamic portfolio optimization problem for a portfolio consisting of 30 stocks forming the German DAX30 stock index from August 2010 to April 2012. We chose the same data set as in  \cite{KilianovaSevcovicANZIAM, KilianovaSevcovicKybernetika} for the purpose of comparison.  As for the drift and volatility function, we will assume their form:
\[
\mu(x,t,\bmtheta) = \bmmu^T\bmtheta - \frac12 \bmtheta^T \bmSigma  \bmtheta  +\varepsilon e^{-x} ,
\quad
\hbox{and} 
\quad 
\sigma(x,t,\bmtheta)^2 =\bmtheta^T \bmSigma  \bmtheta,
\]
where $\bmSigma$ is a positive definite covariance matrix.
The function $\alpha(x,\tau,\varphi)$ can be rewritten as follows: $\alpha(x,\tau,\varphi) = \tilde\alpha(\varphi) -\varepsilon e^{-x}$, 
where $\tilde\alpha$ is the value function of the parametric quadratic optimization problem
\begin{equation}
\tilde\alpha(\varphi) = \min_{ \bmtheta \in \Delta} 
\left(-\bmmu^T\bmtheta +  \frac{\varphi+1}{2} \bmtheta^T \bmSigma  \bmtheta \right)\,.
\label{eq_alpha_def_quadratic}
\end{equation}
A graphical example of the function $\tilde\alpha$ in which $\bmmu$ and $\bmSigma$ were obtained from the DAX30 data set is depicted in Figure \ref{fig:alpha_alphader_alphaderder}. We can observe jumps in the graph of the second derivative of $\alpha$. Indeed, according to Theorem~\ref{smootheness}, the function $\alpha$ is $C^{1,1}$ continuous only. Furthermore, jumps in $\alpha''$ correspond to the points $\varphi$ where the set of indices $\{i : \bmtheta_i>0\}$ with positive weights is enlarged by a new index (cf. \cite{KilianovaSevcovicANZIAM, KilianovaSevcovicKybernetika}).

As for the utility functions, we use
\[
u(x) = -e^{-ax}
\]
for the terminal utility and 
\[
c(x,t) = -\kappa e^{-d x - \varrho (T-t)}
\]
for the intertemporal utility. 

Utility functions parameters used in our example are: $a=9$, $\kappa=1$, $d\in\{0, 8, 11\}$, $\varrho=0$. Parameters corresponding to model data are: $\varepsilon=1$. Parameters of the numerical scheme are: $h=0.01$, $k=0.5h^2$, $x_L=-4$, $x_R=8$, $x_* = x_{200}= -2.01, i_*=200$. Note that the solution does not depend on $x_*$ so one can choose arbitrary $x_*$. Nevertheless, a suitable choice of $x_*$ is very important in order to stabilize numerical computation. The free parameter $x_*$ enters integral term in (\ref{eq_PDEphi-cons}) as well as the ODE for $b$, i.e. (\ref{ODEsystem-cons-b}). We calculate the function $\alpha$ for $\varphi \in (-1,15)$ with a fine division step  $h_{\varphi}=0.05$. The investment period is $T=1$.

\begin{figure}
    \centering
    \includegraphics[width=0.45\textwidth]{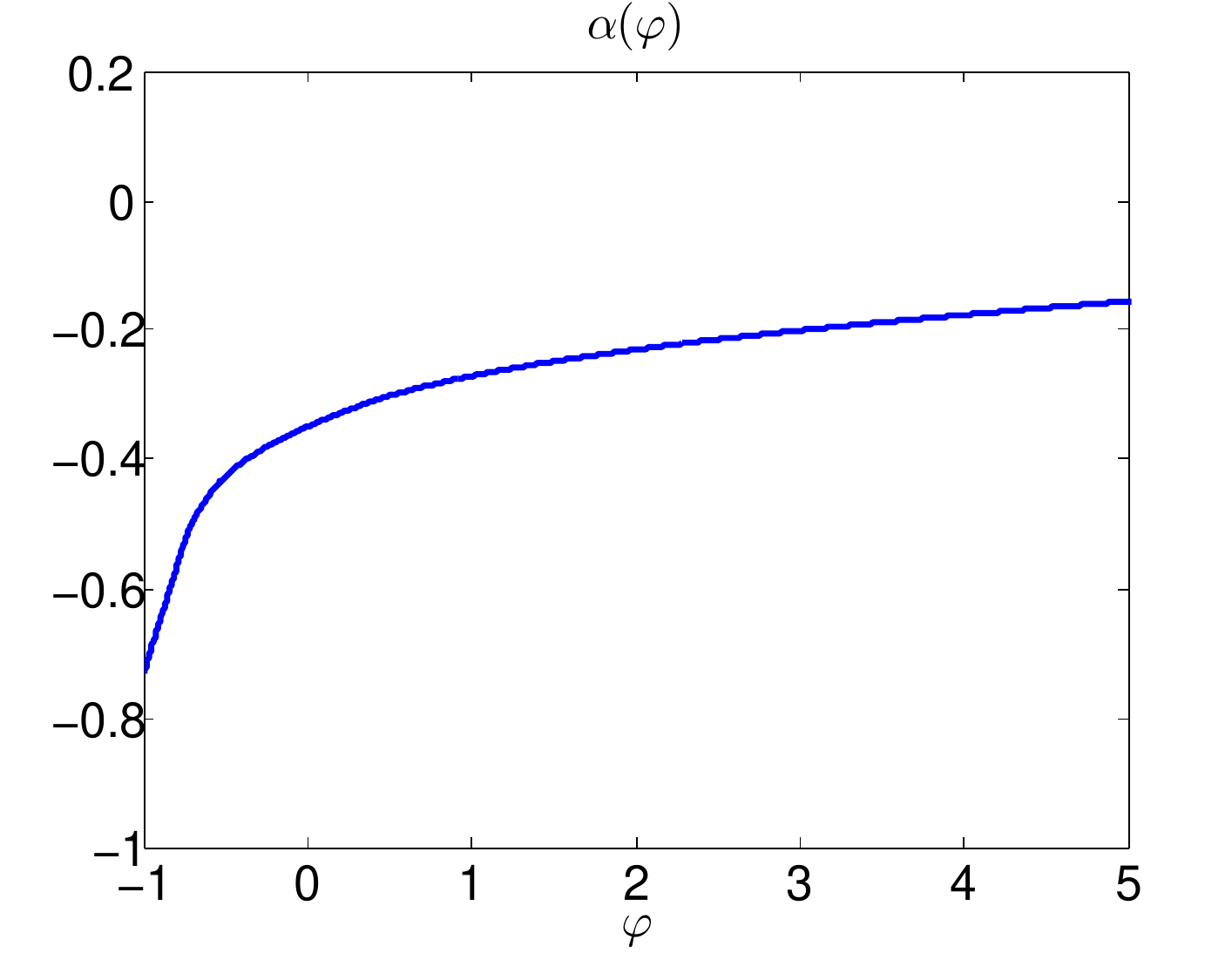}
    \includegraphics[width=0.44\textwidth]{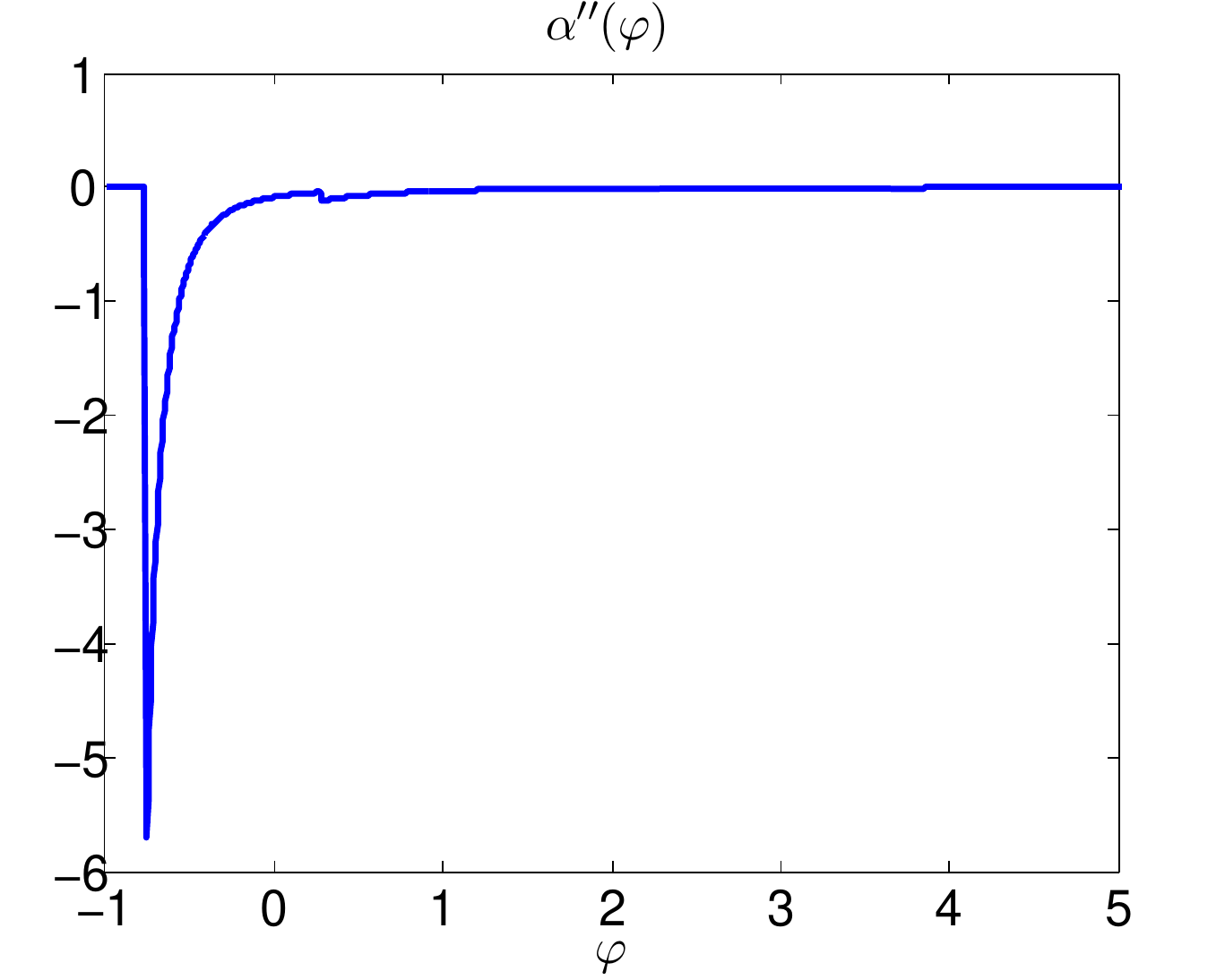}
    \caption{A graph of the value function $\tilde\alpha$ and its second derivative $\tilde\alpha''(\varphi)$ for a portfolio consisting of DAX30 stocks. Source: Kilianov\'a and \v{S}ev\v{c}ovi\v{c}  \cite{KilianovaSevcovicANZIAM, KilianovaSevcovicKybernetika}. }
    \label{fig:alpha_alphader_alphaderder}
\end{figure}

\begin{figure}
    \centering
    \includegraphics[width=0.45\textwidth]{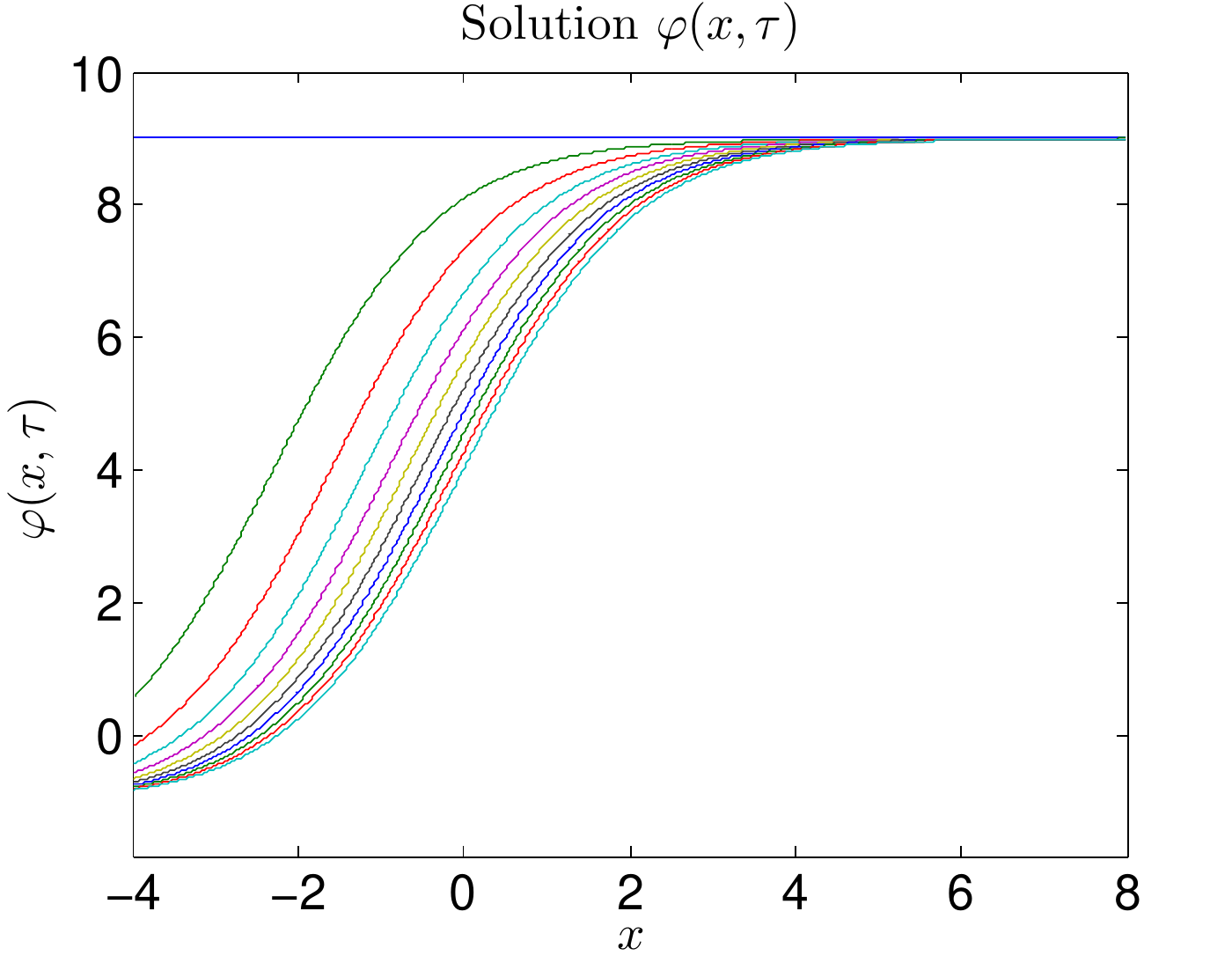}
    \quad
    \includegraphics[width=0.44\textwidth]{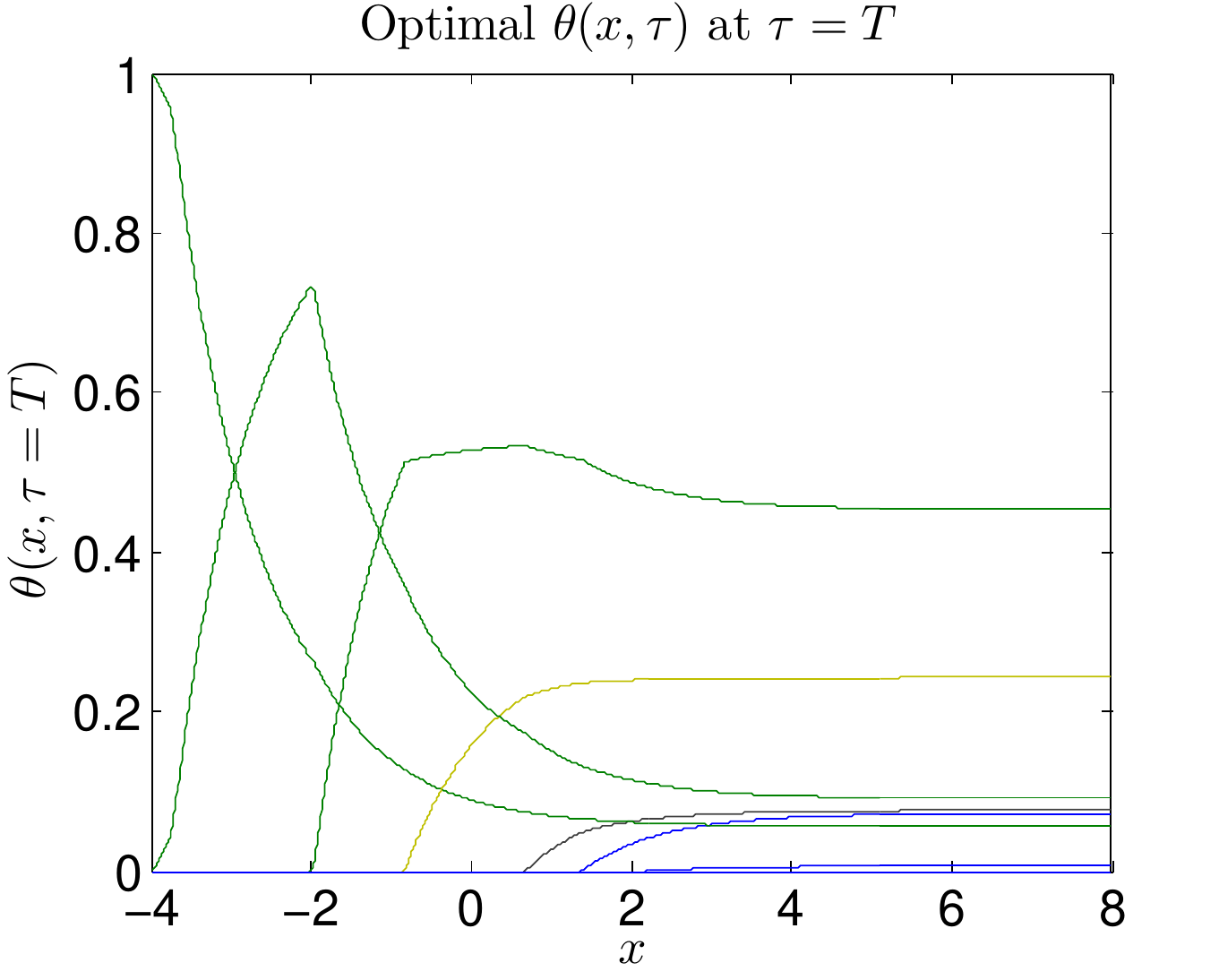} \\    
    \includegraphics[width=0.45\textwidth]{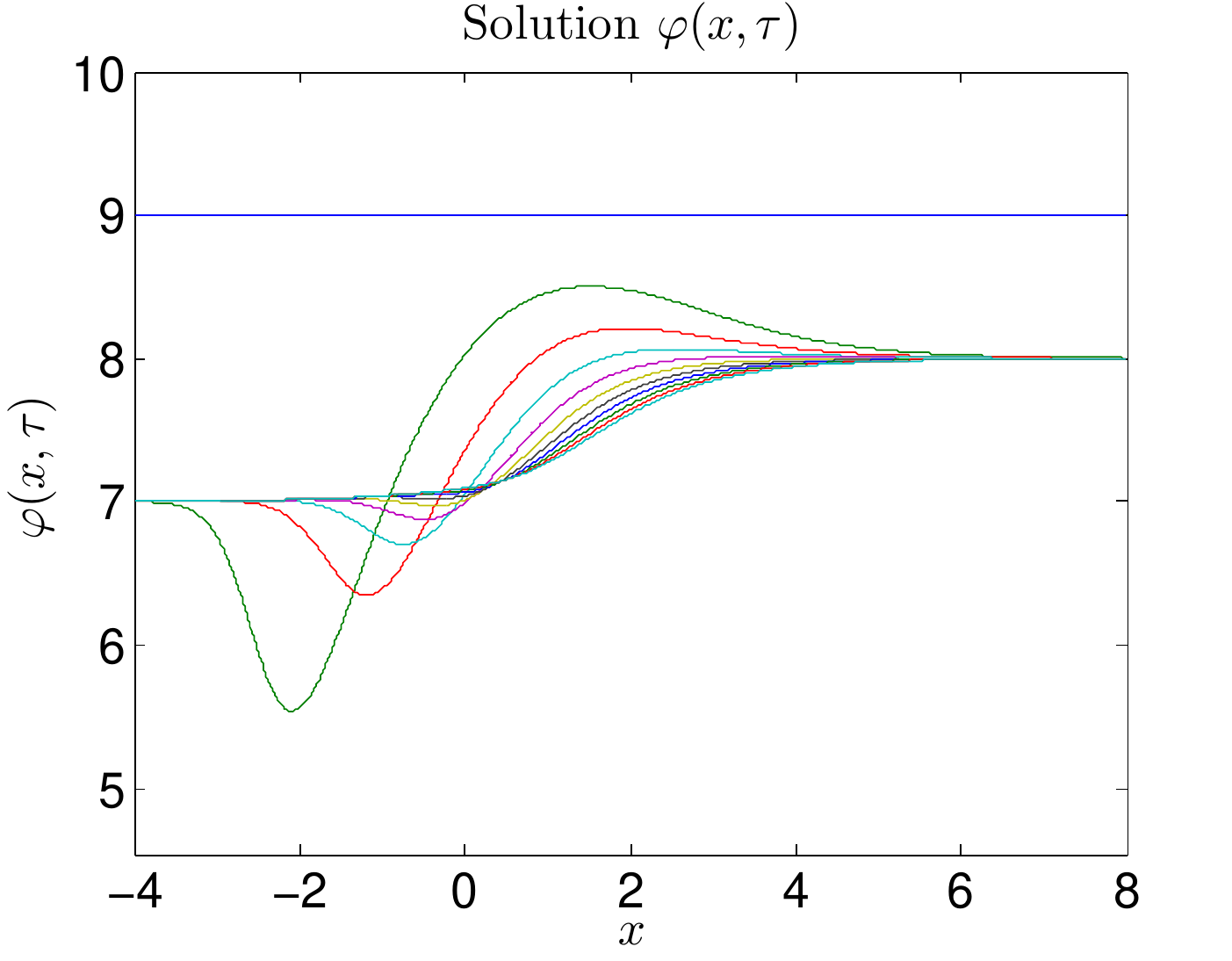}
    \quad
    \includegraphics[width=0.44\textwidth]{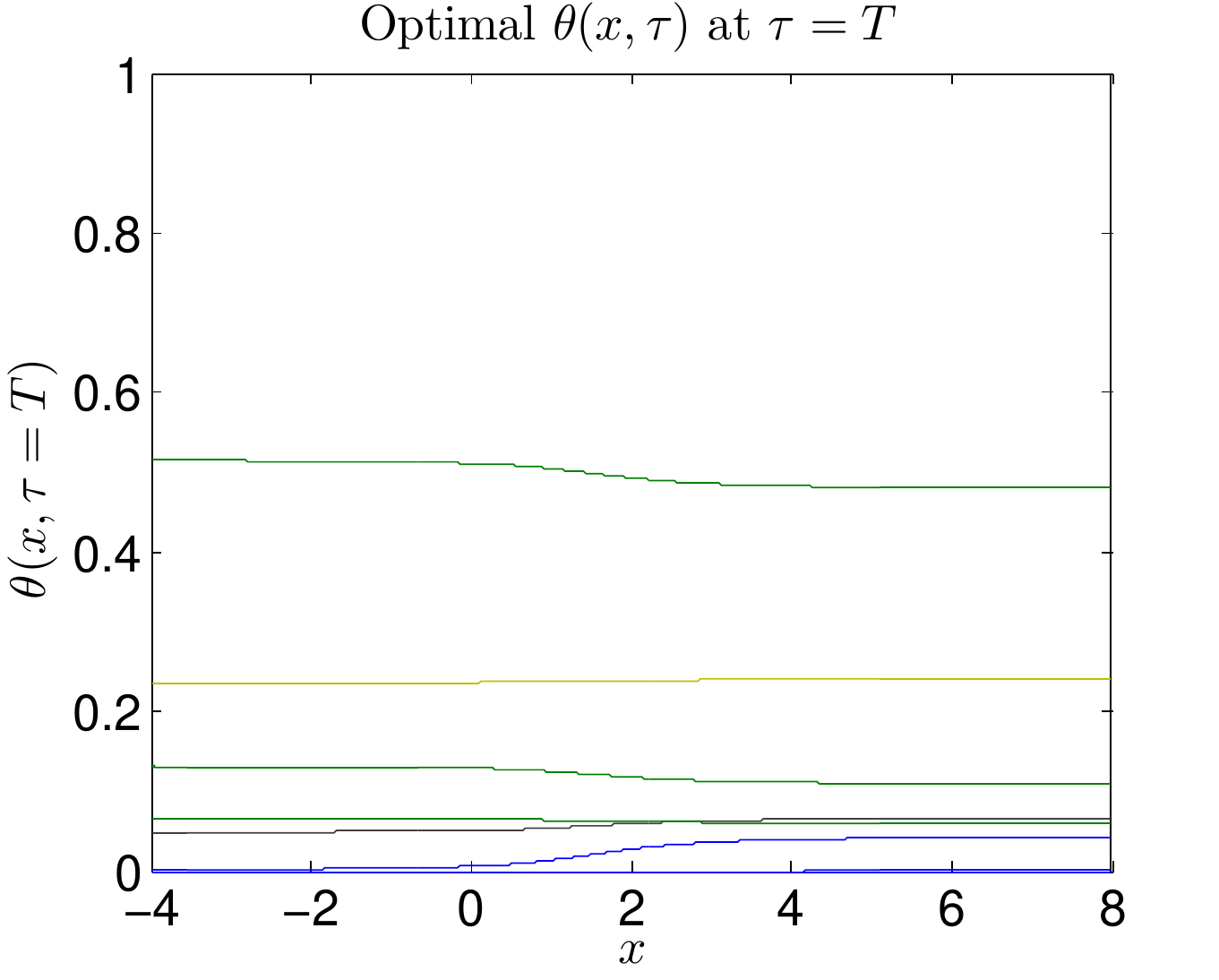} \\  
    \includegraphics[width=0.45\textwidth]{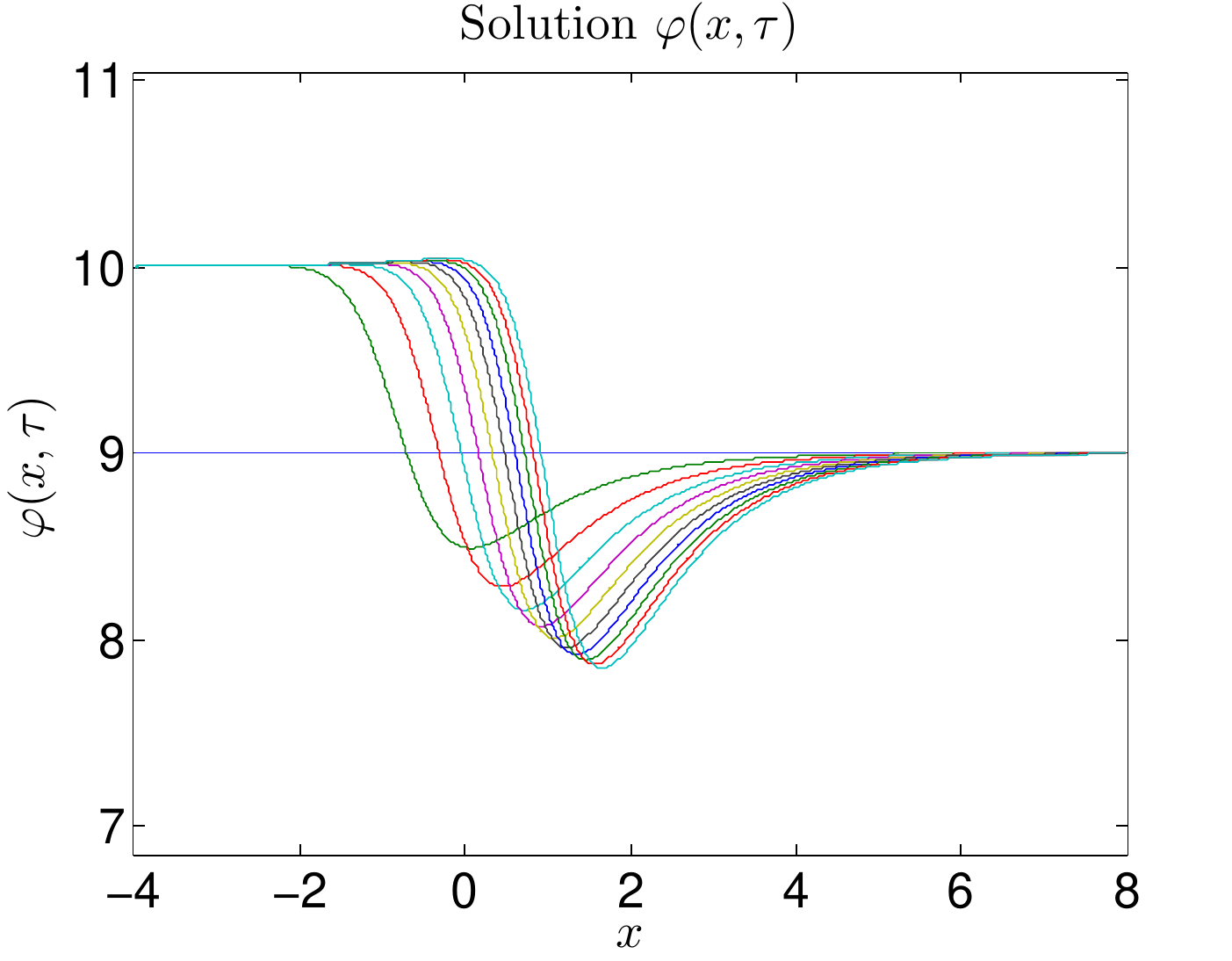}
    \quad
    \includegraphics[width=0.44\textwidth]{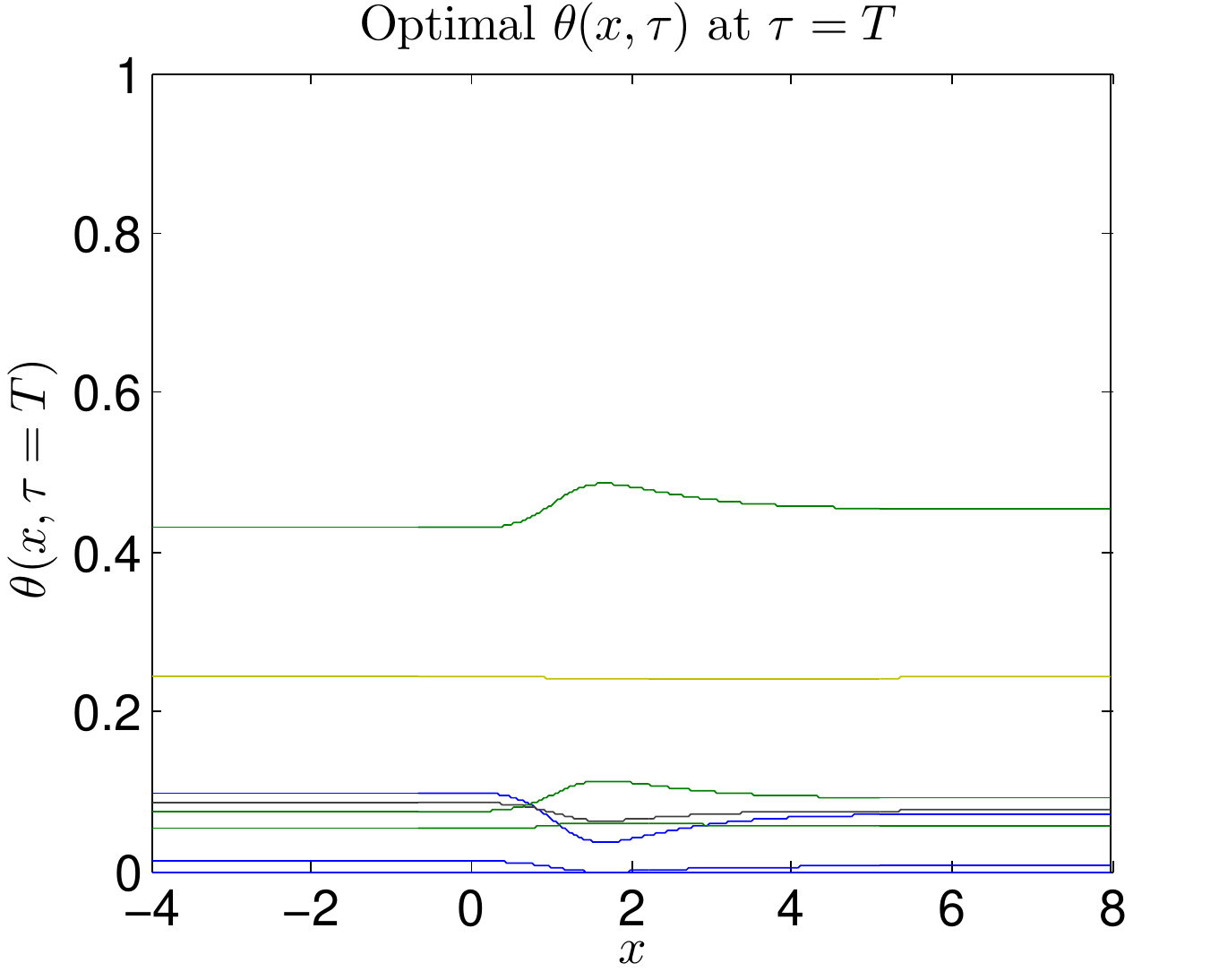}     
    \caption{A solution $\varphi(x,\tau)$ at time instances $jT/10$, $j=0,\cdots,10$, $T=1$, $h=0.01$, $k=0.5h^2$ and optimal portfolio weights $\bmtheta(x,\tau=T)$ for $d=0$ (top left). The constant blue line is the initial condition, then solutions $\varphi(x,\tau^j)$ move from left (green curve) to right for increasing $\tau^j$. Top right plot depicts dependence of active portfolio weights $\theta_i>0$ at $\tau=T$. Next rows correspond to $d=8$ (middle) and $d=11$ (bottom). }
    \label{fig:results_consumption}
\end{figure}

Figure \ref{fig:results_consumption} presents numerical results for $d=0$ (the case of a trivial intertemporal function $c\equiv 0$), $d=8$ and $d=11$. The main difference we can observe is that while for the problem without  intertemporal utility function we obtain a solution $\varphi(x,\tau)$ which is increasing on the interval $[x_L,x_R]$, in problems with a nontrivial intertemporal utility function the solution $\varphi(x,\tau)$ turns out to be non-monotone in $x$. It eventually becomes increasing in the $x$ variable when $\tau$ is approaching the maturity $T$. Furthermore, the range of values of $\varphi$ is a smaller interval when compared to the case without  intertemporal utility function. This has a practical consequence: as $\varphi$ has a small variation in the $x$ variable for $d\approx a$, so does the vector of optimal weights $\bmtheta$ (see Figure~\ref{fig:results_consumption}, right column). Notice that in the case when $\varepsilon=0, c\equiv 0$ there is a constant solution $\varphi(x,\tau)\equiv\varphi(x,0)=a$ to (\ref{eq_PDEphi-cons}) corresponding to the so-called Merton solution to the optimal portfolio selection problem (cf. \cite{Merton2}).  Note  that the solution $\varphi$ satisfies a-priori estimates derived in Theorem~\ref{th-bounds}. 

In summary, there is a non-trivial effect on optimal portfolio selection when considering a non-trival intertemporal utility function $c(x,t)$ which has a similar behavior as the terminal utility function $u(x)$. We furthermore showed that the optimal solution $\varphi(x,\tau)$ to the transformed HJB equation  (\ref{eq_PDEphi-cons}) need not be monotonically increasing. In terms of the optimal portfolio selection vector $\bmtheta$, the optimal weight $\ theta_i$ for some of the stocks entering the active set can attain local minimum with respect to the $x$ variable. Such a behavior cannot be observed in models without intertemporal utility and non-trivial portfolio inflow $\varepsilon>0$ investigated in the recent papers \cite{KilianovaSevcovicANZIAM, KilianovaSevcovicKybernetika}. 

As far as numerical aspects of the Riccati transformation method are concerned, we showed that using this transformation enables to solve a quasi-linear parabolic equation by means of modern numerical methods based on finite volume approximation in a more efficient way when compared to traditional numerical methods based on the  fixed policy iteration method or other explicit numerical approximation approaches the original HJB equation.

\section*{Acknowledgements}
The authors were supported by VEGA 1/0062/18 and DAAD ENANEFA-2018 grants.

\end{document}